\documentclass[a4paper,UKenglish]{llncs}

\usepackage[british]{babel} 
\usepackage[utf8]{inputenc}
\usepackage[T1]{fontenc}

\usepackage{amsmath} 
\usepackage{amsfonts}
\usepackage{amssymb} 
\usepackage{stmaryrd}

\usepackage{url}
\usepackage{hyperref} 
\usepackage{array}
\usepackage{tikz}
\usetikzlibrary{snakes,arrows,arrows.meta,shapes}
\usetikzlibrary{matrix}
\usetikzlibrary{backgrounds}
\usetikzlibrary{patterns}

\usepackage{subcaption}
\usepackage{multirow}
\usepackage{enumitem}
\usepackage{comment}

\usepackage{algorithm}
\usepackage{algpseudocode}

\def\substitute#1#2#3{{#1}_{[\scriptstyle {#2} \mapsto {#3}]}}

\def\powerSet#1{2^{#1}}
\def\singletonSet#1{\{{#1}\}}
\def\valueSet#1#2{\{{#1}, \dots, {#2}\}}
\def\setDescription#1#2{\{{#1} \mid {#2}\}}

\def\infinity{\infty}
\def\minusInfinity{-\infinity}

\def\naturalNumbers{\mathbb{N}}
\def\naturalWithZero{{\naturalNumbers}_{0}}

\def\naturalWithZeroInfinity{{\naturalWithZero} \cup \singletonSet{\infinity}}
\def\naturalWithMinusInfinity{\naturalNumbers \cup \singletonSet{\minusInfinity}}

\def\integer{a}
\def\secondInteger{b}

\def\scalarMultiplication{\cdot}

\def\myValue{k} 
\def\secondValue{q}

\def\absoluteValue#1{|{#1}|}
\def\max{\textsc{max}}
\def\min{\textsc{min}}
\def\sign#1{\textsc{sign}({#1})}

\def\size#1{|{#1}|}

\def\index{i}
\def\secondIndex{j}
\def\thirdIndex{h}
\def\fourthIndex{k} 

\def\vector{a}
\def\secondVector{b}
\def\explicitVector#1#2#3{({#1}_{#2},\dots,{#1}_{#3})}
\def\vectorElement#1#2{{#1}_{#2}}
\def\prefix#1#2{{#1}^{:{#2}}}
\def\suffix#1#2{{#1}^{{#2}:}}
\def\infix#1#2#3{{#1}^{{#2}:{#3}}}

\def\collapse#1{\widetilde{#1}}

\def\orderedCouple#1#2{\langle {#1}, {#2} \rangle}

\def\preset#1{{}^\bullet{#1}}
\def\poset#1{{#1}^\bullet}

\def\define{\stackrel{\Delta}=}
\def\equivalenceDefine{\stackrel{\Delta}\Leftrightarrow}
\def\defineSet#1#2{\define{\setDescription {#1} {#2}}}
\def\assign{:=}

\def\component{v}

\def\componentSet{V}
\def\componentCount{n}

\def\influenceSet{I}

\def\graph{G}
\def\graphTuple{(\componentSet,\influenceSet)}

\def\regulator{u}
\def\regulatorSet#1{n^-(#1)}

\def\regulatorState{\omega}
\def\allRegulatorStateSet{\Omega}
\def\regulatorStateSet#1{{\allRegulatorStateSet}_{#1}}

\def\regulatorProjection#1#2{\omega_{#1}({#2})}

\def\constraintSet{R}

\def\activation{+}
\def\inhibition{-}
\def\observable{\mathrm{o}}

\def\constraintTypeSet{\{{\activation}, {\inhibition}, {\observable}\}}

\def\maximumVector{m}
\def\maximum#1{\vectorElement {\maximumVector} {#1}}
\def\domain#1{\valueSet 0 {\maximum {#1}}}

\def\prn{\graph_\maximumVector^\constraintSet}
\def\prnTuple{(\graph, \maximumVector, \constraintSet)}

\def\state{x}
\def\secondState{y}

\def\stateSet#1{S({#1})}

\def\transition{t}
\def\transitionSet#1{\Delta({#1})}
\def\transitionSubset{T}

\def\transitionValueChange{c}
\def\explicitTransitionValueChange#1#2#3{{#1}_{#2} \rightarrow {#1}_{#3}}

\def\explicitTransitionTuple#1#2{({#1}, {#2})}
\def\transitionTuple{\explicitTransitionTuple{\transitionValueChange}{\regulatorState}}
\def\explicitTransition#1#2#3#4{\explicitTransitionTuple{\explicitTransitionValueChange{#1}{#2}{#3}}{#4}}

\def\fireTransition{\cdot}

\def\componentProjection#1{V({#1})}

\def\transitionSignFunction{s}
\def\transitionSign#1{{\transitionSignFunction}({#1})}

\def\trace{\pi}
\def\secondTrace{\rho}

\def\subtraceInjection{\phi}
\def\inject#1{\subtraceInjection({#1})}

\def\parametrisedNetwork#1#2{({#1}, {#2})}

\def\activationLimitVector{l^A}
\def\inhibitionLimitVector{l^I}

\def\dprnTuple{(\prn, \activationLimitVector, \inhibitionLimitVector)}
\def\dprn{\mathcal{G}}

\def\increase{+1}
\def\decrease{-1}

\def\partialTransition{\delta}
\def\explicitPartialTransition#1#2#3#4{({\explicitTransitionValueChange{#1}{#2}{#3}}, {#4})}
\def\partialTransitionTuple{(\transitionValueChange, \partialRegulatorState)}
\def\fullPartialTransitionTuple{\explicitPartialTransition{\component}{\myValue}{\secondValue}{\partialRegulatorState}}

\def\anyValue{{\ast}}

\def\partialRegulatorState{\aleph}
\def\partialRegulatorStateSet#1{\mathcal{A}_{#1}}
\def\partialRegulatorStateSubset{\mathcal{A}}

\def\coverSetLimitFunction#1{\sigma_{#1}}
\def\coverSetLimit#1#2{\coverSetLimitFunction{#1}({#2})}

\def\regulationCoverSetBase{\mathcal{A}}
\def\regulationCoverSet#1{{\regulationCoverSetBase}_{#1}}

\def\parameterOrder{\unlhd}

\def\parametrisation{P}
\def\parametrisationSet{\mathcal{P}}
\def\allParametrisations#1{\mathbb{P}({#1})}

\def\parameter#1#2#3{\vectorElement{#1}{{#2}, {#3}}}

\def\admissibleSetFunction{\Psi}
\def\admissibleSet#1{\admissibleSetFunction({#1})}

\def\concreteSemantics{{\admissibleSetFunction}_{C}}
\def\abstractSemantics{{\admissibleSetFunction}_{A}}

\def\lowerBoundParametrisation{L}
\def\upperBoundParametrisation{U}
\def\parametrisationLattice{({\lowerBoundParametrisation},
  {\upperBoundParametrisation})}

\def\poLEQ{\leq}
\def\poGEQ{\geq}
\def\parametrisationOrder{\poLEQ}

\def\goalComponent{g}
\def\goalValue{\top}
\def\goal{{\goalComponent}_{\goalValue}}

\def\objectiveTransitionSet#1{\tau({#1})}
\def\validObjectiveTransitionSet#1#2{\tau_{#1}({#2})}

\def\objective{O}
\def\explicitObjective#1#2#3{{#1}_{#2} \rightsquigarrow {#1}_{#3}}
\def\objectiveValues#1#2{{#1} \rightsquigarrow {#2}}

\def\objectiveSet{\mathcal{O}}
\def\reducedObjectiveSet{\mathcal{B}}

\def\objectiveSignFunction{s}
\def\objectiveSign#1{{\objectiveSignFunction}({#1})}

\title{Combining Refinement of Parametric Models with Goal-Oriented Reduction of Dynamics%
\thanks{
This work has been partly funded by
ANR-FNR project ``AlgoReCell'' ANR-16-CE12-0034,
by Labex DigiCosme (project ANR-11-LABEX-0045-DIGICOSME) operated by ANR as part of the program
``Investissement d'Avenir'' Idex Paris-Saclay (ANR-11-IDEX-0003-02),
and by
ERATO HASUO Metamathematics for Systems Design Project (No. JPMJER 1603), JST.
}}
\author{Stefan Haar\inst{1}, Juraj Kol\v{c}\'ak\inst{1,2}, Lo\"ic Paulev\'e\inst{3,4}}
\institute{
	LSV, CNRS \& ENS Paris-Saclay, Universit\'e Paris-Saclay, France
    \and
    National Institute of Informatics, Tokyo, Japan
    \and
LRI UMR 8623, Univ. Paris-Sud -- CNRS, Universit\'e Paris-Saclay, Orsay, France
\and
Univ. Bordeaux, Bordeaux INP, CNRS, LaBRI, UMR5800, F-33400 Talence, France}

\begin{document}

\maketitle

\begin{abstract}
    Parametric models abstract part of the specification of dynamical models by integral parameters.
    They are for example used in computational systems biology, notably with parametric regulatory
    networks, which specify the global architecture (interactions) of the networks, while
    parameterising the precise rules for drawing the possible temporal evolutions of the states of
    the components.
    A key challenge is then to identify the discrete parameters corresponding to concrete models with
    desired dynamical properties.
    This paper addresses the restriction of the abstract execution of parametric regulatory
    (discrete) networks by the means of static analysis of reachability properties (goal states).
Initially defined at the level of concrete parameterised models, the goal-oriented reduction of
dynamics is lifted to parametric networks, and is proven to preserve all the minimal traces to the
specified goal states.
It results that one can jointly perform the refinement of parametric networks (restriction of domain
of parameters) while reducing the necessary transitions to explore and preserving reachability
properties of interest.
\end{abstract}


\section{Introduction}

Various cyber and physical systems are studied by the means of discrete dynamical models which
describe the possible temporal evolution of the state of the components of the system.
Defining such models requires extensive knowledge on the underlying system for specifying the rules
which generate the admissible state transitions over time.
Usually, and especially for physical systems, such as biological networks for which discrete models
are extensively employed \cite{Thomas73,Helikar2012,Cohen2015,Bartocci2016,Collombet2017,ColomotoNotebook2018},
it is common to lack such precise knowledge, making an accurate specification of discrete models challenging.

With \emph{parametric} models, part of the specification of the rules for
generating the discrete transitions is encoded as (integral) parameters.
Thus, a parametric model abstracts a set of concrete \emph{parameterised} models, this set being
characterised by the domain of parameter values.

In this paper, we focus on
\emph{Parametric Regulatory Networks} (PRNs), also known as Thomas Networks
\cite{Thieffry95,Bernot07,Bernot17,me18},
which are commonly employed for modelling qualitative dynamics of biological systems.
PRNs allow separating biological knowledge on the pairwise interactions (the architecture of the network) from
the rules of interplay between the interactions, usually less known.

In the literature, PRNs are mainly used as a basic framework for identifying fully parameterised
models (i.e., Boolean and multilevel networks) which satisfy dynamical
properties typically generated from experimental data.
This identification task, related to so-called model inference
and process mining~\cite{Carmona16,Mokhov16,Koutny16,deLeon15},
consists in transforming an abstract parametric model into a set of concrete parameterised
models verifying desired dynamical properties.
For PRNs, state-of-the-art methods rely on parameter enumeration~\cite{smbionet},
coloured model-checking~\cite{me12},
logic programming and Boolean satisfiability~\cite{Corblin2012,Ostrowski16},
and Hoare logic~\cite{Bernot18-Hoare}.

However, the exhaustive identification of parameters is often limited to small models, as the set of
parameterised models can turn out to be too large to be exhaustively enumerated and further
analysed.

In \cite{me18}, we introduced a semantics of PRNs enabling the \emph{refinement} of a PRN by
restricting the domain of its parameters without having to enumerate concrete models,
keeping them in a compact abstract representation instead.
The refinement is performed according to concrete discrete state transitions: the domain of parameters is
restricted so that it abstracts all the concrete models in which the state transition is
admissible.
Essentially, such semantics of PRNs enable efficient exploration of dynamics of a set of
parameterised models.

This exploration suffers from the same bottleneck as individual parameterised models: the number of
reachable states grows exponentially with the number of components and thus, becomes intractable for
large networks.
The exploration of the reachable state space is usually performed to verify dynamical properties.
Consequently, various \emph{model reduction} methods have been designed on concrete parameterised models
to enhance the tractability of their verification \cite{HP-ppl06,Talcott2006,Pauleve17,Chatain17}:
by reducing the transitions to consider, these methods limit the reachable state space to explore
while guaranteeing the correctness of the verification.

In this paper, we address the combination of refinement operations on parametric models with
model reductions initially defined at the level of concrete individual models.
Essentially, the challenge consists in lifting up such model reductions so they can be performed at
the abstract level of parametric models, while ensuring the correctness of their refinement.

We focus on reachability properties, i.e., starting in an initial configuration, the ability to
eventually reach a given (partial) configuration.
On the one hand, we are interested in refining PRNs to accurately identify concrete
parameterised discrete network models that verify the reachability property;
on the other hand, we want to take advantage of goal-driven exploration of dynamics of parameterised
models to ignore transitions which do not influence the reachability of the goal, enhancing the
tractability of the analysis.

The refinement of PRNs we consider for reachability properties has been introduced in \cite{me18}.
It consists in dynamically drawing transitions allowed by at least one concrete model, and
subsequently restricts the domain of parameters to exclude models which do not allow the drawn
transition.
The generation of transitions is done directly from the abstract representation of the set of concrete
models, and therefore involves no enumeration of parameterised models.

The goal-oriented model reduction we consider has been introduced in \cite{Pauleve17} at the level of
parameterised network models.
Given a reachability property (goal), the method relies on static analysis by abstract
interpretation to identify transitions which are not involved in any minimal trace leading to the
goal. Here the minimality refers to the absence of a sub-trace.
Whereas deciding reachability properties in parametrised models (namely automata networks)
is a PSPACE-complete problem~\cite{ChengEP95},
the goal-oriented model reduction has a complexity polynomial in the number of
components and exponential in the in-degree of components in the networks
(components having a direct influence on a single one).

In this article we present a lifting of the goal-oriented reduction from
parametrised models to sets of models with shared architecture,
represented by parametric models.
To this end, we introduce a directed version of PRNs which allow us to efficiently
capture model reduction without the need to explicitly enumerate all possible transitions.
We conduct the reduction itself on abstract dynamics of PRNs where instead of
enumerating all enabled transitions, we only consider the minimal necessary
condition for each component to change value.

The introduced reduction method can be applied on-the-fly to speed up
reachability checking in parametric models.
Thanks to the preservation of \emph{all} minimal traces, it is guaranteed to capture all
parametrised models capable of reproducing the coveted behaviour.

\paragraph*{Outline}
Section~\ref{sec:prns} recalls the definition of parametric regulatory networks,
their dynamics, constraints on influences and finally presents a generalised
parametrisation set semantics. In Section~\ref{sec:reduction}, the goal-oriented
model reduction procedure is extended from parametrised models
to parametric models. Directed version of PRNs is introduced for this purpose
alongside an abstraction of dynamics designed to alleviate
the reduction complexity.
Section~\ref{sec:regulation_cover_set_inference} supplies an algorithm for computing
a suitable abstraction of PRN dynamics used in the reduction procedure.
Finally, Section~\ref{sec:discussion} summarises the results and offers a brief
introduction to possible extensions and applications
as well as future directions.

\paragraph*{Notations}
We use $\prod$ to build Cartesian products between sets.
As the ordering of components matters, $\prod$ is not commutative.
Therefore, we write
$\prod^{\leq}_{x\in X}$
for the product over elements in $X$ according to a total order $\leq$.
To ease notations, when the order is clear from the context, or when either $X$ is a set of
integers, or a set of integer vectors, on which we use the lexicographic
ordering, we simply write $\prod_{x\in X}$.
Given a sequence of $n$ elements
$\pi = (\pi_i)_{1 \leq i \leq n}$,
we write $\collapse{\pi}\defineSet{\pi_i}{1\leq i\leq n}$ for the set of its elements.
Given a vector $v=\langle v_1,\dots,v_n\rangle$,
we write $\substitute{v}{i}{y}$ for the vector equal to $v$ except on the component $i$, which is equal to
$y$.


\section{Parametric Regulatory Networks}
\label{sec:prns}

Regulatory networks are finite discrete dynamical systems where the components evolve
individually with respect to the value of (a few) other components, their regulators.
The value of components in regulatory networks ranges in a finite discrete domain, usually
represented as $\{0, ..., m\}$ for some $m\in\mathbb N$, thus extending Boolean networks~\cite{Thomas73}.
The evolution of components is then defined by discrete functions which associate to the global states
of the network the value towards which each component tends.

Thus, defining regulatory networks requires knowledge on which components influence each
others, and how the value of each component is computed from the value of its regulators.
\emph{Parametric} Regulatory Networks (PRNs) allow to decouple this specification by having on the
one hand a fixed architecture of the network, so-called \emph{influence graph}, and on the other
hand discrete parameters, which when instantiated specify the functions of the regulatory network.


\subsection{Influence graph and constraints}
\label{sec:influence_graph_and_constraints}

The influence graph encodes the directed interactions between the components
of the regulatory networks:
a component $u$ having a direct influence on component $v$ means that in some states of the regulatory network,
the computation of the value of node $v$ \emph{may} depend on the value of $u$.
Importantly, if the component $w$ has no direct influence of $v$, then the computation of the value
on $v$ never depends on the value of $w$.


{ 

\begin{definition}[Influence Graph]
\label{def:influence_graph}
An \emph{influence graph} $\graph$ is a tuple $\graphTuple$
where $\componentSet$ is a finite set of $\componentCount$ nodes (components)
and $\influenceSet \subseteq \componentSet \times \componentSet$
is a set of directed edges (influences).

\noindent
For each $\component \in \componentSet$ we denote the set of its \emph{regulators} by
$\regulatorSet{\component} \defineSet{\regulator \in \componentSet}
{(\regulator,\component) \in \influenceSet}$.

\end{definition}
}

Besides the existence/absence of direct influences between components, it is usual to have some
knowledge about the nature of the influences.
Two kinds of constraints are generally considered: \emph{signs} and \emph{observability}.

Influence signs are captured by monotonicity constraints.
An influence $(\regulator, \component) \in \influenceSet$ is
\emph{positive-monotonic}, denoted $\activation$,
if the sole increase of the value of the regulator
$\regulator$ cannot cause a decrease of the computed value of the target $\component$.
Symmetrically, an influence $(\regulator, \component) \in \influenceSet$
is \emph{negative-monotonic}, denoted $\inhibition$,
if the sole increase in the value of $\regulator$
cannot cause an increase in the computed value of $\component$.
An influence $(\regulator, \component) \in \influenceSet$ is
\emph{observable}, denoted $\observable$,
if there exists a state in which the sole change of the value
of $\regulator$ induces a change of the computed value of $\component$.
Thus observability enforces that $\regulator$ does have an influence on the value of $\component$,
in some states of the regulatory network.
Remark that observability does not imply positive/negative monotonicity --
e.g., when the value of $\component$ is computed as the exclusive disjunction XOR between its own
value and the value of $\regulator$.

Let us denote a set of influence constraints for an influence graph $\graphTuple$ as
$\constraintSet\subseteq \influenceSet \times \constraintTypeSet$.
An example of influence constraint set is given in
Figure~\ref{fig:prn_example} (a) as labels on edges of the influence graph.


\subsection{Parametrisation}
\label{sec:parametrisation}

Let us consider an influence graph $\graph=\graphTuple$ among $\componentCount$ components
and a set of influences constraints $\constraintSet$.
Let us denote by $\maximumVector \in {\naturalNumbers}^{\componentCount}$
the vector specifying the maximum discrete value of each component:
the states of the regulator networks span $\prod_{\component\in\componentSet}\domain{\component}$.
The computation of the value of each component of a regulatory network
is constrained by $\graph$, $\constraintSet$ and $\maximumVector$.
In particular, $\graph$ imposes that the value of a component depends only on its regulators.

A \emph{regulator state} $\regulatorState$ of a component
$\component \in \componentSet$ is a vector specifying the value of each
regulator of $\component$.
We denote the set of all regulator states of a component as
$\regulatorStateSet{\component} \define
\prod_{\regulator \in \regulatorSet{\component}} \domain{\regulator}$.
Intuitively, a regulator state of a component $\component$ is a projection of
a global state of the network (states of all components)
to just the regulators of $\component$, that fully determine its evolution.

A \emph{parameter} $\orderedCouple{\component}{\regulatorState}$
then represents a target value towards which component
$\component \in \componentSet$ evolves in regulator state
$\regulatorState \in \regulatorStateSet{\component}$.
We denote the set of all parameters as
$\allRegulatorStateSet \define \bigcup_{\component \in \componentSet}
\singletonSet{\component} \times \regulatorStateSet{\component}$

A \emph{parametrisation} $\parametrisation$ is a vector assigning a value
to each parameter. The set of all parametrisations associated to an influence graph $\graph$
and a maximum value vector $\maximumVector$ is therefore given by
$\allParametrisations{\graph_\maximumVector} =
\prod_{\orderedCouple{\component}{\regulatorState} \in \allRegulatorStateSet}
^{\parameterOrder} \domain{\component}$
where $\parameterOrder$ is an arbitrary, but fixed total order on parameters.
The set of all parametrisations satisfying both the influence graph
$\graph=\graphTuple$ and influence constraints $\constraintSet$ with maximum value vector $\maximumVector$ is
then defined as:
\def\someParametrisation{\parametrisation \in \allParametrisations{\prn}}
\def\someRegState{\regulatorState \in \regulatorStateSet{\component}}
\def\someValue{\myValue \in \valueSet{1}{\maximum{\regulator}}}
\def\valueParameter{\parameter{\parametrisation}{\component}{\substitute{\regulatorState}{\regulator}{\myValue}}}
\def\minusOneParameter{\parameter{\parametrisation}{\component}{\substitute{\regulatorState}{\regulator}{\myValue - 1}}}
\begin{align*}
\allParametrisations{\prn}
 \define \{ \parametrisation \in
\textstyle\prod_{\orderedCouple{\component}{\regulatorState} \in \allRegulatorStateSet}
^{\parameterOrder} & \domain{\component}
 \mid \forall u,v\in\componentSet, \\
 (\regulator,\component,\activation) \in \constraintSet & \Rightarrow
\forall \someRegState,\forall \someValue:
\valueParameter \geq \minusOneParameter \\
(\regulator,\component,\inhibition)\in\constraintSet &\Rightarrow
\forall \someRegState,\forall \someValue:
\valueParameter \leq \minusOneParameter \\
(\regulator,\component,\observable) \in\constraintSet&\Rightarrow
\exists \someRegState,\exists \someValue:
\valueParameter \neq \minusOneParameter
 \}
\end{align*}

\subsection{Parametric Regulatory Networks}

A \emph{Parametric Regulatory Network} (PRN) gathers an influence graph $\graph$,
influence constraints $\constraintSet$, and maximum value vector $\maximumVector$, to which can then
be associated a subset of parametrisations $\allParametrisations{\prn}$.
A (\emph{parametrised}) regulatory network can then be defined by a couple $(\prn,\parametrisation)$
where $\parametrisation\in\allParametrisations{\prn}$.


{ 
\def\valueChange{\explicitTransitionValueChange{\component}{\index}{\secondIndex}}
\def\someTransition{\explicitTransitionTuple{\valueChange}{\regulatorState}}

\begin{definition}
\label{def:parametric_regulatory_network}

A \emph{parametric regulatory network} (PRN) is a tuple $\prnTuple$, written $\prn$,
where $\graph$ is an influence graph between $\componentCount$ components,
$\constraintSet$ is a set of influence constraints,
and $\maximumVector\in\mathbb{N}^\componentCount$ is a vector of
the maximum values of each component.

\noindent
The set of \emph{states} of $\prn$ is denoted by
$\stateSet{\prn} \define \prod_{\component\in\componentSet}\domain{\component}$.

\noindent
The set of \emph{(local) transitions} of $\prn$ is denoted by:
$\transitionSet{\prn}\define \{\someTransition \mid
\component \in \componentSet \wedge $
$\regulatorState \in \regulatorStateSet{\component} \wedge
\index, \secondIndex \in \domain{\component} \wedge
\absoluteValue{\index - \secondIndex} = 1 \}$.
We use $\componentProjection{\someTransition} = \component$
to denote the component whose value is changed by transition
$\someTransition \in \transitionSet{\prn}$.
Furthermore, $\transitionSign{\someTransition} = \transitionSign{\valueChange} =
\secondIndex - \index$ denotes the sign of the transition (value change).

\end{definition}

A transition $\someTransition \in \transitionSet{\prn}$ is \emph{enabled}
in state $\state \in \stateSet{\prn}$
if $\vectorElement{\state}{\component} = \index$
and $\regulatorProjection{\component}{\state} = \regulatorState$, where
$\regulatorProjection{\component}{\state}$ is the projection of state $\state$
to the regulators of $\component$.
Given a state $\state$ and a transition $\transition$ enabled in $\state$,
$\state \fireTransition \transition$ denotes the state
$\substitute{\state}{\component}{\secondIndex} \in \stateSet{\prn}$
obtained by firing transition $\transition$ in $\state$.

Finally, a transition
$\explicitTransition{\component}{\index}{\secondIndex}{\regulatorState} \in
\transitionSet{\prn}$ is \emph{enabled} by a parametrisation set
$\parametrisationSet\subseteq\allParametrisations{\prn}$
if there exists a  parametrisation $\parametrisation \in \parametrisationSet$ such that the parameter value
$\parameter{\parametrisation}{\component}{\regulatorState} =
\secondIndex + \integer \scalarMultiplication \transitionSign{
\explicitTransitionValueChange{\component}{\index}{\secondIndex}}$
for some $\integer \in \naturalWithZero$.



{ 

\begin{example}
\label{exp:parametric_regulatory_network}

An example of a PRN $\prn$ composed of an influence graph $\graph=\graphTuple$ and
vector $\maximumVector = {\singletonSet{1}}^{\size{\componentSet}}$ is depicted
in Figure~\ref{fig:prn_example}.
Based on the influences in $\graph$ and maximum values in $\maximumVector$,
all regulator states of each component,
which correspond to parameters of $\prn$, are determined.
The table in Figure~\ref{fig:prn_example} (b) lists all the parameters alongside
an example parametrisation $\parametrisation \in \allParametrisations{\prn}$.
$\prn$ combined with $\parametrisation$ identifies a unique parametrised network
$(\prn, \parametrisation)$.
The dynamics of $(\prn, \parametrisation)$ are given in
Figure~\ref{fig:prn_example} (c).

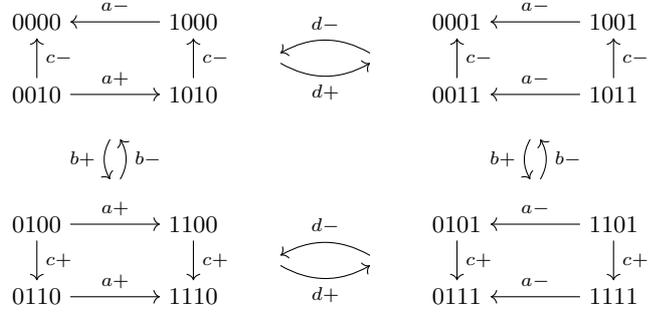
\begin{figure}[th!]
\centering
\begin{subfigure}{.42\linewidth}
  \centering

{ 

\begin{tikzpicture}
  \tikzstyle{every node}=[draw,ellipse,minimum width=20pt,minimum height=20pt,font=\bf\small,align=center];
  \tikzset{edge/.style={-{Stealth[scale=2,length=3,width=2]},line width=.5pt}};

  \node[] (a) {$a$};
	\node[right=50pt, above=30pt] (b) {$b$};
  \node[right=50pt] (c) {$c$};
	\node[right=50pt, below=30pt] (d) {$d$};

  \draw[edge, in=20, out=120, looseness=6] (b) to
  node[draw=none, above, midway] {$\inhibition\observable$} (b);
	\draw[edge, in=330, out=240, looseness=6] (d) to
  node[draw=none, below, midway] {$\inhibition\observable$} (d);
	\draw[edge, bend left=10] (b) to
  node[draw=none, right, midway] {$\activation\observable$} (c);
	\draw[edge, bend right=25] (b) to
  node[draw=none, above, midway] {$\activation$} (a);
  \draw[edge, bend left=25] (d) to
  node[draw=none, below, midway] {$\inhibition\observable$} (a);
  \draw[edge] (c) to
  node[draw=none, above, midway] {$\activation\observable$} (a);

\end{tikzpicture}
}
  \caption{The influence graph $\graph$.}
\end{subfigure}
\begin{subfigure}{.42\linewidth}

  \begin{tabular}{ll|c}
    \hline
    $\parameter{\parametrisation}{a}{\langle b=0,c=0,d=0\rangle}$
     & \multirow{8}{*}{$\in \domain{a}$} & 0\\
    $\parameter{\parametrisation}{a}{\langle b=1,c=0,d=0\rangle}$ & & 1\\
    $\parameter{\parametrisation}{a}{\langle b=0,c=1,d=0\rangle}$ & & 1\\
    $\parameter{\parametrisation}{a}{\langle b=1,c=1,d=0\rangle}$ & & 1\\
    $\parameter{\parametrisation}{a}{\langle b=0,c=0,d=1\rangle}$ & & 0\\
    $\parameter{\parametrisation}{a}{\langle b=1,c=0,d=1\rangle}$ & & 0\\
    $\parameter{\parametrisation}{a}{\langle b=0,c=1,d=1\rangle}$ & & 0\\
    $\parameter{\parametrisation}{a}{\langle b=1,c=1,d=1\rangle}$ & & 0\\\hline
    $\parameter{\parametrisation}{b}{\langle b=0\rangle}$ & \multirow{2}{*}{$\in \domain{b}$} & 1\\
    $\parameter{\parametrisation}{b}{\langle b=1\rangle}$ & & 0\\\hline
    $\parameter{\parametrisation}{c}{\langle b=0\rangle}$ & \multirow{2}{*}{$\in \domain{c}$} & 0\\
    $\parameter{\parametrisation}{c}{\langle b=1\rangle}$ & & 1\\\hline
    $\parameter{\parametrisation}{d}{\langle d=0\rangle}$ & \multirow{2}{*}{$\in \domain{d}$} & 1\\
    $\parameter{\parametrisation}{d}{\langle d=1\rangle}$ & & 0\\\hline
  \end{tabular}

  \caption{All the parameters of PRN $\prn$ with an example parametrisation
  $\parametrisation$.}
\end{subfigure}
\begin{subfigure}{.8\linewidth}

{ 

\begin{tikzpicture}[ampersand replacement=\&]
\matrix (m) [matrix of math nodes,row sep=0.5em,column sep=1.5em] {
   \&      \&    \&      \&    \& {} \&    \&      \&    \&      \&    \\
   \& 0000 \&    \& 1000 \&    \&    \&    \& 0001 \&    \& 1001 \&    \\
   \&      \&    \&      \& {} \&    \& {} \&      \&    \&      \&    \\
   \& 0010 \&    \& 1010 \&    \&    \&    \& 0011 \&    \& 1011 \&    \\
   \&      \& {} \&      \&    \&    \&    \&      \& {} \&      \&    \\
{} \&      \&    \&      \&    \&    \&    \&      \&    \&      \& {} \\
   \&      \& {} \&      \&    \&    \&    \&      \& {} \&      \&    \\
   \& 0100 \&    \& 1100 \&    \&    \&    \& 0101 \&    \& 1101 \&    \\
   \&      \&    \&      \& {} \&    \& {} \&      \&    \&      \&    \\
   \& 0110 \&    \& 1110 \&    \&    \&    \& 0111 \&    \& 1111 \&    \\
};


\path[->,every node/.style={font=\scriptsize}]
(m-5-3) edge[bend right] node[left] {$b+$} (m-7-3)
(m-7-3) edge[bend right] node[right] {$b-$} (m-5-3)

(m-5-9) edge[bend right] node[left] {$b+$} (m-7-9)
(m-7-9) edge[bend right] node[right] {$b-$} (m-5-9)

(m-3-5) edge[bend right] node[below] {$d+$} (m-3-7)
(m-3-7) edge[bend right] node[above] {$d-$} (m-3-5)

(m-9-5) edge[bend right] node[below] {$d+$} (m-9-7)
(m-9-7) edge[bend right] node[above] {$d-$} (m-9-5)





(m-8-2) edge node [right] {$c+$} (m-10-2)
(m-8-4) edge node [right] {$c+$} (m-10-4)
(m-8-8) edge node [right] {$c+$} (m-10-8)
(m-8-10) edge node [right] {$c+$} (m-10-10)

(m-4-2) edge node [right] {$c-$} (m-2-2)
(m-4-4) edge node [right] {$c-$} (m-2-4)
(m-4-8) edge node [right] {$c-$} (m-2-8)
(m-4-10) edge node [right] {$c-$} (m-2-10)

(m-4-2) edge node[above] {$a+$} (m-4-4)
(m-8-2) edge node[above] {$a+$} (m-8-4)
(m-10-2) edge node[above] {$a+$} (m-10-4)

(m-2-4) edge node[above] {$a-$} (m-2-2)
(m-2-10) edge node[above] {$a-$} (m-2-8)
(m-4-10) edge node[above] {$a-$} (m-4-8)
(m-8-10) edge node[above] {$a-$} (m-8-8)
(m-10-10) edge node[above] {$a-$} (m-10-8)
;
\end{tikzpicture}
}
  \caption{States and transitions of $(\prn, \parametrisation)$ depicted as
  nodes and edges of a state space graph respectively.
  Since components $b$ and $d$ may update values independently of other
  components (i.e. in any state),
  their value changes are only displayed schematically to improve readability.}
\end{subfigure}

\caption{Influence graph with influence constraints as labels, parameters and dynamics of a possible parametrisation
of PRN $\prn$.\label{fig:prn_example}}

\end{figure}

\end{example}
}

}

We capture the basic semantics of a PRN by traces,
which correspond to different possible behaviours of the network.


{ 
\def\length{\size{\trace}}
\def\traceElement#1{\vectorElement{\trace}{#1}}
\def\fireTrace#1{\state \fireTransition \traceElement{1} \fireTransition \dots \fireTransition \traceElement{#1}}

\begin{definition} [Trace]
\label{def:prn_trace}
Given a PRN $\prn$ and set of parametrisations
$\parametrisationSet\subseteq\allParametrisations{\prn}$,
a finite sequence $\trace = \explicitVector{\trace}{1}{\length}$
of transitions in $\transitionSet{\prn}$ is a \emph{trace} of $\prn$
starting in state $\state \in \stateSet{\prn}$
iff $\forall \index \in \valueSet{1}{\length}:
\traceElement{\index}$ is enabled in state $\fireTrace{\index - 1}$
and by parametrisation set $\parametrisationSet$.

To simplify notation, we use
$\preset{\trace} = \state$ and
$\poset{\trace} = \fireTrace{\length}$.
Moreover, let $\prefix{\trace}{\index} = \explicitVector{\trace}{1}{\index}$,
$\suffix{\trace}{\index} = \explicitVector{\trace}{\index}{\size{\trace}}$ and
$\infix{\trace}{\index}{\secondIndex} =
\explicitVector{\trace}{\index}{\secondIndex}$ denote the \emph{prefix},
\emph{suffix} and \emph{infix} sub-traces of $\trace$ respectively.

\end{definition}
}

With $\parametrisation\in\allParametrisations{\prn}$
and $\parametrisationSet = \{\parametrisation\}$, the above definition gives the traces of the
parametrised regulatory network $(\prn,\parametrisation)$.
In the general case, each transition is independently enabled with respect to any parametrisation in
$\parametrisationSet$.


\subsection{Parametrisation Set Semantics}
\label{sec:parametrisation_set_semantics}

With the basic PRN semantics (Definition~\ref{def:prn_trace}),
two transitions are allowed to fire consequently in a single trace
despite no single parametrisation enabling both of them.
To forbid such behaviours, semantics have been introduced for PRNs
that associate each trace (set of transitions)
with a set of parametrisations~\cite{me18}.
Said trace can then be extended only by transitions enabled under some
parametrisation from the associated set.

The purpose of parametrisation set semantics is to discriminate
transitions based on their causal history.
This is done by progressive restriction of the set of parametrisations
to \emph{admissible} parametrisations.
Following~\cite{me18}, we consider a parametrisation $\parametrisation$ to be admissible
if all transitions in the set (causal history) are enabled under $\parametrisation$.
However, we allow for a more lenient definition to make room for over-approximation.


{ 
\def\traceElement#1{\vectorElement{\trace}{#1}}
\def\otherSubset{{\transitionSubset}'}

\begin{definition}[Parametrisation Set Semantics]
\label{def:parametrisation_set_semantics}
Given a PRN $\prn$, a function $\admissibleSetFunction: \powerSet{\transitionSet{\prn}} \rightarrow
\powerSet{\allParametrisations{\prn}}$ is a \emph{parametrisation set semantics}
of $\prn$ iff:
\begin{enumerate}
  \item $\forall \transitionSubset \subseteq \transitionSet{\prn}:
    \setDescription{\parametrisation \in \allParametrisations{\prn}}
      {\forall \transition \in \transitionSubset: \parametrisation \ \text{enables}\ \transition}
    \subseteq \admissibleSet{\transitionSubset}$,
  \item $\forall \transitionSubset, \otherSubset \subseteq \transitionSet{\prn}:
  \transitionSubset \subseteq \otherSubset \Rightarrow
  \admissibleSet{\otherSubset} \subseteq \admissibleSet{\transitionSubset}$.
\end{enumerate}

A trace $\trace$ of the PRN $\prn$ is \emph{realisable} according to the
parametrisation set semantics if and only if
$\admissibleSet{\collapse{\trace}} \neq \emptyset$.

\end{definition}
}

A best abstraction $\concreteSemantics$ producing parametrisation sets of exactly
all the parametrisations that allow each transition in the input set has been defined in \cite{me18}.
To facilitate practical application, as the number of parametrisations may be
in the worst case double exponential in the number of components,
\cite{me18} has tackled the semantics $\abstractSemantics$ over-approximating
parametrisation sets by convex covers, keeping track of only a maximal and minimal element
and thus avoiding the need to enumerate parametrisations explicitly.
More formally, let us first reintroduce the \emph{parametrisation order}.


{ 

\begin{definition}
\label{def:parametrisation_order}

The \emph{parametrisation order} on vectors of length $\myValue$
is the partial order $\parametrisationOrder$ defined as follows:

\[
	\vector \poLEQ \secondVector \equivalenceDefine \forall \index \in
	\valueSet{0}{\myValue - 1}: \vectorElement{\vector}{\index} \leq
	\vectorElement{\secondVector}{\index}
\]

\end{definition}
}

The parametrisation set given by $\abstractSemantics$ is a couple of parametrisations
$\parametrisationLattice$ representing the lower and upper bound.
Formally, $\parametrisationLattice =
\setDescription{\parametrisation \in \allParametrisations{\prn}}
{\lowerBoundParametrisation \poLEQ \parametrisation \wedge
\upperBoundParametrisation \poGEQ \parametrisation}$
is a bounded convex sublattice of all vectors of length
$\size{\allRegulatorStateSet}$ with the parametrisation order.
In~\cite{me18}, a method has been provided to compute
\emph{the tightest lower and upper bounds} for a given set of transitions
and influence constraints without the need to explicitly enumerate the
parametrisations.

Naturally, checking whether a particular parametrisation belongs to the abstracted set
can be done simply by comparing it with the bounds.
Similarly, determining whether a transition is enabled (by any parametrisation)
can be done without explicit enumeration of the parametrisations.
In fact, it is enough to compare against the corresponding parameter value
of the relevant bound, e.g.
$\vectorElement{\upperBoundParametrisation}{\component, \regulatorState}
\geq \myValue + 1$ is the sufficient and necessary condition for the transition
$\explicitTransition{\component}{\myValue}{\myValue + 1}{\regulatorState}$
to be enabled.

\medskip

In this article, we consider any parametrisation set semantics compatible
with Definition~\ref{def:parametrisation_set_semantics}, however,
special attention is given to $\abstractSemantics$ as it can be used with the restriction method
without the need to enumerate the parametrisations explicitly.


\section{Goal-Oriented Reduction}
\label{sec:reduction}

In this section, we extend the goal-oriented model reduction procedure from
parametrised models (in particular, automata networks)~\cite{Pauleve17}
to the parametric models.


\subsection{Minimal Traces}
\label{sec:minimal_traces}

Given a PRN $\prn$ and a state $\state \in \stateSet{\prn}$,
we say a value $\goalValue \in \domain{\goalComponent}$
of a component $\goalComponent \in \componentSet$ is \emph{reachable}
from $\state$ iff either $\vectorElement{\state}{\goalComponent} = \goalValue$
or there exists a realisable trace $\trace$ with $\preset{\trace} = \state$
and $\vectorElement{\poset{\trace}}{\goalComponent} = \goalValue$.

We are interested in reachability by \emph{minimal} traces.
Adapted from~\cite{Pauleve17}, a realisable trace is minimal for $\goal$
reachability if there exists no other realisable trace reaching $\goal$
with a subsequence of transitions.


{ 
\def\parametrisedPRN{(\prn, \parametrisation)}
\def\smallerTrace{\rho}
\def\traceDom{\valueSet{1}{\size{\trace}}}
\def\smallerTraceDom{\valueSet{1}{\size{\smallerTrace}}}
\def\traceElement#1{\vectorElement{\trace}{#1}}
\def\smallerTraceElement#1{\vectorElement{\smallerTrace}{#1}}

\begin{definition} [Minimal Trace]
\label{def:minimal_trace}
Given a parametrised PRN $\parametrisedPRN$, a trace $\trace$
of $\parametrisedPRN$ is \emph{minimal} w.r.t. reachability of goal
$\goal$ from state $\state$ if and only if there exists no other
trace $\smallerTrace$ satisfying $\state = \preset{\smallerTrace}$,
$\vectorElement{\poset{\smallerTrace}}{\goalComponent} = \goalValue$,
$\size{\smallerTrace} < \size{\trace}$ and existence of an injection
$\subtraceInjection: \smallerTraceDom \rightarrow \traceDom$
such that $\forall \index,\secondIndex \in \smallerTraceDom:
\index \leq \secondIndex \Rightarrow \inject{\index} \leq \inject{\secondIndex}$
and $\smallerTraceElement{\index} = \traceElement{\inject{\index}}$.

\end{definition}
}

An important property of minimal traces is their independence on the exact
parametrisation. More precisely, using parametrisation set semantics,
if a trace is minimal for at least one parametrisation,
then it is minimal for any other parametrisation under which it is enabled.


{ 
\def\otherParametrisation{{\parametrisation}'}
\def\discreteNetwork{\parametrisedNetwork{\prn}{\parametrisation}}
\def\otherDiscreteNetwork{\parametrisedNetwork{\prn}{\otherParametrisation}}
\def\traceParamSet{\admissibleSet{\collapse{\trace}}}
\def\otherTraceParamSet{\admissibleSet{\collapse{\secondTrace}}}

\begin{property}[Parametrisation Independence of Minimal Traces]
\label{ppt:minimal_trace_independence}
Let $\prn$ be a PRN and $\trace$ a realisable trace minimal in
$\discreteNetwork$ for some $\parametrisation \in \traceParamSet$.
Then, $\trace$ is minimal in any $\otherDiscreteNetwork$
where $\otherParametrisation \in \traceParamSet$

\end{property}

\begin{proof}

$\otherParametrisation \in \traceParamSet$ guarantees $\trace$ is a proper trace
of $\otherDiscreteNetwork$.
We conduct the rest of the proof by contradiction.
Let thus $\secondTrace$ be a trace in $\otherDiscreteNetwork$
satisfying the conditions in Definition~\ref{def:minimal_trace}.
From the existence of the injection $\subtraceInjection$ we get
$\collapse{\secondTrace} \subseteq \collapse{\trace}$ and from the definition
of parametrisation set semantics $\traceParamSet \subseteq \otherTraceParamSet$.
$\secondTrace$ is therefore realisable in $\discreteNetwork$
meaning that $\trace$ is not minimal in $\discreteNetwork$
which is a contradiction.
\end{proof}
}

Property~\ref{ppt:minimal_trace_independence} allows us to speak of a realisable
trace of a PRN as minimal without the need to explicitly state the
parametrisation which is witness to the minimality.


\subsection{Directed Parametric Regulatory Networks}
\label{sec:directed_parametric_regulatory_networks}

The goal-oriented reduction for parametrised models is facilitated
by pruning transitions which are guaranteed to not be used by any
minimal trace reaching the goal~\cite{Pauleve17}.
The PRN definition could be extended to allow for pruning the transition
set to subsets $\transitionSubset \subseteq \transitionSet{\prn}$.
Unlike the case of general parametrised models, however,
the transitions of a PRN only allow to change the value of a component by steps of size $1$.
As such, if a transition increasing the value of a component
$\component \in \componentSet$ to $\myValue \in \domain{\component}$
is to be pruned, all transitions increasing the value of $\component$
beyond $\myValue$ can surely be pruned as well,
and symmetrically for decreasing transitions.
Thus, instead of removing individual transitions of PRNs, we disable
increasing, respectively decreasing, value of a component
in a given regulator state beyond a certain value (or entirely).
This is facilitated by keeping record of the activation (increase)
and inhibition (decrease) limits for each component in vectors
$\activationLimitVector$ and $\inhibitionLimitVector$ respectively.


{ 
\def\source{\state}
\def\target{\secondState}
\def\someRegState{\regulatorState \in \allRegulatorStateSet}

\begin{definition} [Directed Parametric Regulatory Network]
\label{def:directed_prn}
A \emph{directed parametric regulatory network} (DPRN) is a tuple
$\dprn=\dprnTuple$, where $\prn$ is a parametric regulatory network,
$\activationLimitVector \in {(\naturalWithMinusInfinity)}^
{\size{\allRegulatorStateSet}}$
is a vector of activation limits for each regulator state $\someRegState$
and $\inhibitionLimitVector \in {(\naturalWithZeroInfinity)}^
{\size{\allRegulatorStateSet}}$
is a vector of inhibition limits for each regulator state $\someRegState$.

The set of states of $\dprn$ is equal to the set of states
of the underlying PRN: $\stateSet{\dprn} = \stateSet{\prn}$.

The set of transitions of $\dprn$ is a subset of the PRN transitions satisfying
the activation and inhibition limits
$\activationLimitVector$ and $\inhibitionLimitVector$ respectively.
Formally, $\transitionSet{\dprn} \subseteq \transitionSet{\prn}$ such that:
$$\forall \transition
 = \explicitTransition{\component}{\index}{\secondIndex}{\regulatorState}
  \in \transitionSet{\prn}:
\transition \in \transitionSet{\dprn} \equivalenceDefine
\begin{cases}
  \index < \vectorElement{\activationLimitVector}{\regulatorState}
   &\text{if} \ \transitionSign{\transition} = \increase \\
  \index > \vectorElement{\inhibitionLimitVector}{\regulatorState}
   &\text{if} \ \transitionSign{\transition} = \decrease
\end{cases}$$

\end{definition}
}

One may remark that by using parametrisation set semantics, it is already
possible to restrict the activation or inhibition of components in individual
regulator states while just using PRNs.
While it is true that an equivalent set of enabled transitions
can be achieved both by restricting the parametrisation set and by DPRN,
the semantics of the two restrictions are different.

The parametrisation set semantics serves primarily to keep track of
parametrisations capable of reproducing certain behaviour(s),
and thus restrict the set of enabled transitions based on their causal history.
On the other hand, the $\activationLimitVector$ and $\inhibitionLimitVector$
of DPRN mark components whose activation or inhibition (beyond a certain value)
is not necessary to reach a given goal (via a minimal trace).
A parametrisation that allows changing a component value beyond the limit,
thus allowing behaviour which does not lead to the established goal may still
allow a different sequence of transitions leading to the goal.
We want to retain such parametrisations, thus the "useless" behaviour which
does not lead to the goal cannot be restricted
in the parametrisation set semantics.
Therefore, keeping the information about parametrisations and about
the activation and inhibition limits independently is key.

The complete independence of parametrisation set semantics and the limit vectors
$\activationLimitVector$ and $\inhibitionLimitVector$ allows us to employ both
in parallel.
The extension of both traces (Definition~\ref{def:prn_trace})
and parametrisation set semantics
(Definition~\ref{def:parametrisation_set_semantics}) from PRNs to DPRNs
is thus natural.


\subsection{Objectives}

The reduction for parametrised models relies on identifying sub-goals,
or \emph{objectives}, local in terms of individual components.
We reintroduce the concept of a (local) objective for the parametric model.


{ 
\def\someObjective{\explicitObjective{\component}{\index}{\secondIndex}}
\def\traceIndex{k}

\begin{definition}[Objective]
\label{def:objective}
Given an DPRN $\dprn$, an \emph{objective}
$\someObjective$ is a pair of values
$\index, \secondIndex \in \domain{\component}$ of a component
$\component \in \componentSet$.

An objective $\someObjective$ is \emph{valid} in a starting state
$\state \in \stateSet{\dprn}$ iff $\index = \secondIndex$
or a realisable trace $\trace$ of the parametrised DPRN exists,
such that $\preset{\trace} = \state$,
$\vectorElement{\poset{\trace}}{\component} = \secondIndex$
and $\exists \traceIndex \in \valueSet{0}{\size{\trace} - 1}:
\vectorElement{\preset{\vectorElement{\trace}{\traceIndex}}}{\component}
 = \index$.

$\objectiveValues{\index}{\secondIndex}$ is used to denote $\someObjective$
if the component $\component \in \componentSet$ is obvious from the context.

Each objective $\someObjective$ captures either increase or decrease
of the value of the component.
Formally, the sign of an objective
$\objectiveSign{\someObjective} \define \sign{\secondIndex - \index}$.

\end{definition}
}

By requiring the witness of objective validity to be a realisable trace
instead of just a trace of enabled transitions, we retain only behaviours
which are present in at least one parametrised model.

The objective represents a change of value of only one component
$\component \in \componentSet$.
A realisable trace reproducing such a change may, however,
require to also change value of other components,
namely the regulators of $\component$.
Each objective is thus associated to a set of transitions which may be used
to complete it, and from which the required regulator values can be obtained.


\subsection{Regulation Cover Sets}
\label{sec:regulation_cover_sets}

Depending on the parametrisation set semantics, it may be a common occurrence
for a particular value change to be enabled by numerous regulator states
(recall that enabling is existential w.r.t. parametrisations).
Such cases lead to a substantial redundancy in individual transition
enumeration as the value of only a subset of regulators may be enough to
determine whether a value change is enabled or not.
To this end we introduce a definition of a partial regulator state,
which is used to represent a (minimal) condition for a value change to be enabled.


{ 
\def\regulatorInPartial{\vectorElement{\partialRegulatorState}{\regulator}}

\begin{definition}[Partial Regulator State]
\label{def:partial_regulator_state}
A \emph{partial regulator state} $\partialRegulatorState$ of component
$\component \in \componentSet$ is a vector
$\partialRegulatorState \in \prod_{\regulator \in \regulatorSet{\component}}
\domain{\regulator} \cup \{\anyValue\}$ assigning a value or a wildcard
character $\anyValue$ to each regulator $\regulator$ of $\component$.
By abuse of notation, $\partialRegulatorState$ is also a set of
regulator states, more precisely
$\partialRegulatorState \subseteq \regulatorStateSet{\component}$ such that
for all $\regulatorState \in \regulatorStateSet{\component}$:
$$\regulatorState \in \partialRegulatorState \equivalenceDefine
\forall \regulator \in \regulatorSet{\component}:
\vectorElement{\regulatorState}{\regulator} = \regulatorInPartial
 \vee \regulatorInPartial = \anyValue$$

The set of all partial regulator states of $\component \in \componentSet$
is denoted as $\partialRegulatorStateSet \component$.

\end{definition}
}

Partial regulator states can be utilised to abstract the DPRN dynamics
while minimising the number of repetitions of each value of each regulator.
We capture these abstractions by the means of sets of partial regulator states,
called \emph{regulation cover sets}, representing the enabling condition of
a given value change.
We impose two conditions on regulation cover set of value change
$\transitionValueChange$.
First, the set has to cover all regulator states $\regulatorState$
such that $\transitionTuple$ is enabled.
In other words, for each such regulator state there must exist one or more
partial regulator states which specify the value of each regulator
in $\regulatorState$.
Second, no bad regulator state $\regulatorState$
such that $\transitionTuple$ is not enabled is subsumed by any of the partial
regulator states in the cover set.
These two conditions not only guarantee that the abstract dynamics enable
exactly the same value changes as the concrete dynamics,
but also preserve the regulator information, i.e. each value of each regulator
that appears in the enabling conditions.
The regulator information is necessary to accurately determine which
regulator values are necessary to complete an objective.


{ 
\def\coverSet{\regulationCoverSet{\transitionValueChange}}
\def\limitFunction{\coverSetLimitFunction{\transitionValueChange}}
\def\limit#1{\coverSetLimit{\transitionValueChange}{#1}}
\def\badPartialRegulatorStateSet{{\textsc{inv}}_{\transitionValueChange}({\partialRegulatorStateSubset})}
\def\mergeFunction{{\mu}_{\transitionValueChange}}
\def\merge#1{\mergeFunction({#1})}

\begin{definition}[Regulation Cover Set]
\label{def:regulation_cover_set}
Let $\dprn$ be a DPRN and $\parametrisationSet$ a parametrisation set
from the parametrisation set semantics,
and let $\transitionValueChange
 = \explicitTransitionValueChange{\component}{\index}{\secondIndex}$
be an arbitrary value change of a component $\component \in \componentSet$.
A set of partial regulator states
$\coverSet \subseteq \partialRegulatorStateSet{\component}$
is a cover set of $\transitionValueChange$ iff the following is satisfied:
\begin{itemize}
  \item $\forall \regulatorState \in \regulatorStateSet{\component}:
  \transitionTuple$ is enabled under $\parametrisationSet$:
  $\forall \regulator \in \regulatorSet{\component}:
  \exists \partialRegulatorState \in \coverSet:
  \regulatorState \in \partialRegulatorState \wedge
  \vectorElement{\regulatorState}{\regulator} =
  \vectorElement{\partialRegulatorState}{\regulator}$.
  \item $\forall \regulatorState \in \regulatorStateSet{\component}:
  \transitionTuple$ is not enabled:
  $\forall \partialRegulatorState \in \coverSet:
  \regulatorState \notin \partialRegulatorState$.
\end{itemize}

\end{definition}
}

Any regulation cover set, including the concrete regulation cover set
$\setDescription{\regulatorState}{\regulatorState \in
\regulatorStateSet{\component}: \transitionTuple \ \text{is enabled}}$,
may be used for the purposes of the reduction procedure.
The aim of the regulation cover set being to minimise the number of individual
regulator values which appear across all of the partial regulator states,
an algorithm that computes regulator cover sets with no more regulator value specifications
than the concrete regulation cover set is introduced in Section~\ref{sec:regulation_cover_set_inference}.


\subsection{Reduction of Directed Parametric Regulatory Networks}
\label{sec:dprn_reduction}

Our reduction procedure essentially relies on associating to objectives the set
of (partial) transitions which are necessary to realise the objective
within the corresponding components of the PRN.
Starting from the final (goal) objective, the procedure then recursively
collects objectives related to the identified transitions.

Since PRNs allows only unitary value changes,
the realisation of an objective
$\explicitObjective{\component}{\index}{\secondIndex}$ involves
a monotonic change of value of component $\component$ from $\index$
to $\secondIndex$, where each change of value depends on specific (partial)
regulator state.
This coupling of a value change with a corresponding partial regulator state
is referred to as a \emph{partial transition}.


{ 
\def\someObjective{\explicitObjective{\component}{\index}{\secondIndex}}
\def\regulatorInPartial{\vectorElement{\partialRegulatorState}{\regulator}}
\def\someValueChange{\explicitTransitionValueChange{\component}{\myValue}{\secondValue}}
\def\somePartialTransition{\explicitPartialTransition{\component}{\myValue}{\secondValue}{\partialRegulatorState}}

\begin{definition}[Objective Transition Set]
\label{def:objective_transitions}
Let $\dprn$ be an DPRN parametrised by $\parametrisationSet$,
and let $\someObjective$ be an objective for $\component \in \componentSet$.
The \emph{objective transition set} $\objectiveTransitionSet{\someObjective}$
is defined as
$\objectiveTransitionSet{\someObjective} \define\emptyset$ whenever $i=j$, otherwise,
\begin{align*}
    \objectiveTransitionSet{\someObjective} \define
    \{\somePartialTransition &\mid
  \transitionSign{\someValueChange} = \objectiveSign{\someObjective} \wedge
  \partialRegulatorState \in
  \regulationCoverSet{\someValueChange}
  \\ & \quad
  \wedge
  \max{\{\myValue, \secondValue\}} \leq \max{\{\index, \secondIndex\}}
  \wedge
  \min{\{\myValue, \secondValue\}} \geq \min{\{\index, \secondIndex\}}
  \}
\end{align*}

\noindent
Given an initial state $\state \in \stateSet{\dprn}$, the
\emph{valid objective transition set} of an objective $\someObjective$
in state $\state$ is a subset of the objective transition set
$\validObjectiveTransitionSet{\state}{\someObjective}
 \subseteq \objectiveTransitionSet{\someObjective}$
such that:
$\partialTransitionTuple \in \validObjectiveTransitionSet{\state}{\someObjective}
 \equivalenceDefine \forall \regulator \in \regulatorSet{\component}:
\regulatorInPartial \neq \anyValue \Rightarrow
\objectiveValues{\vectorElement{\state}{\regulator}}{\regulatorInPartial}$
is valid in state $\state$.

\noindent
The (valid) objective transition sets extend to sets of objectives in
the natural manner:
$\objectiveTransitionSet{\objectiveSet} = \bigcup_{\objective \in \objectiveSet}
\objectiveTransitionSet{\objective}$.

\end{definition}
}

Remark that the definition of a valid objective transition set benefits
from the use of partial regulator states.
Indeed, instead of having to check validity of an objective for each regulator,
only the minimal necessary subset of regulators is considered.
Checking objective validity consists of searching for a realisable trace,
which translates to finding all possible extensions (enabled transitions)
of a trace.
As enabled transitions can be retrieved using $\abstractSemantics$
without explicitly enumerating the parametrisations, the validity check is
compatible with $\abstractSemantics$.

The goal-oriented reduction of DPRNs can then be defined by recursively
collecting objectives from partial transitions ($\reducedObjectiveSet$)
and refining the component activation and inhibition limits accordingly.


{ 
\def\regulatorInPartial{\vectorElement{\partialRegulatorState}{\regulator}}

\begin{definition}[Reduction Procedure]
\label{def:reduction_procedure}
The goal-oriented reduction of a \newline DPRN $\dprn=\dprnTuple$
for an initial state $\state \in \stateSet{\dprn}$ and a goal $\goal$
is the DPRN ${\dprn}' = (\prn, {\activationLimitVector}', {\inhibitionLimitVector}')$
with ${\activationLimitVector}'$ and ${\inhibitionLimitVector}'$
being defined as follows, $\forall \component \in \componentSet,
\forall \regulatorState \in \regulatorStateSet{\component}$:
\begin{align*}
  \vectorElement{{\activationLimitVector}'}{\regulatorState} =&
  \max (\setDescription{\myValue \in \domain{\component}}
  {\exists \explicitPartialTransition{\component}{\myValue - 1}{\myValue}
  {\partialRegulatorState} \in
  \validObjectiveTransitionSet{\state}{\reducedObjectiveSet}:
  \regulatorState \in \partialRegulatorState} \cup \singletonSet{\minusInfinity}) \\
  \vectorElement{{\inhibitionLimitVector}'}{\regulatorState} =&
  \min (\setDescription{\myValue \in \domain{\component}}
  {\exists \explicitPartialTransition{\component}{\myValue + 1}{\myValue}
  {\partialRegulatorState} \in
  \validObjectiveTransitionSet{\state}{\reducedObjectiveSet}:
  \regulatorState \in \partialRegulatorState} \cup \singletonSet{\infinity})
\end{align*}
where $\reducedObjectiveSet$ is the smallest set of objectives satisfying the following:
\begin{enumerate}
	\item $\objectiveValues{\vectorElement{\state}{\goalComponent}}{\goalValue}
	 \in \reducedObjectiveSet$
	\item $\forall \objective \in \reducedObjectiveSet:
	\forall \fullPartialTransitionTuple \in
	\validObjectiveTransitionSet{\state}{\objective}:
	\forall \regulator \in \regulatorSet{\component} \setminus
	\singletonSet{\component}: \regulatorInPartial \neq \anyValue \Rightarrow
	\objectiveValues{\vectorElement{\state}{\regulator}}{\regulatorInPartial} \in
	\reducedObjectiveSet$
	\item $\forall \objective \in \reducedObjectiveSet:
	\forall \fullPartialTransitionTuple \in
	\validObjectiveTransitionSet{\state}{\objective}:
	\forall \explicitObjective{\component}{\index}{\secondIndex} \neq
	\objective \in \reducedObjectiveSet:
	\explicitObjective{\component}{\secondValue}{\secondIndex} \in
	\reducedObjectiveSet$\enspace.
\end{enumerate}

\end{definition}
}


{ 

\begin{example}
\label{exp:reduction_procedure}

Consider the parametric regulatory network $\prn$ introduced in
Example~\ref{exp:parametric_regulatory_network} converted to a DPRN
$\dprn=\dprnTuple$ in an unrestrictive manner
($\activationLimitVector = {\singletonSet{1}}^{\componentSet}$ and
$\inhibitionLimitVector = {\singletonSet{0}}^{\componentSet}$),
and a parametrisation set containing only two parametrisations
$\parametrisationSet=\{\parametrisation, {\parametrisation}'\}$, where
$\parametrisation$ is the parametrisation from
Example~\ref{exp:parametric_regulatory_network} and ${\parametrisation}'$
differs from $\parametrisation$ only in value of
$\parameter{{\parametrisation}'}{a}{\langle b=1,c=0,d=0 \rangle} = 0$.
Furthermore, let $a=1$ be a goal and $\state = \langle a=0,b=0,c=0,d=0 \rangle$ an initial state.

In Figure~\ref{fig:state_space_graph} we recall the dynamics of $\dprn$
given as a state space graph. Note that the second parametrisation
${\parametrisation}'$ is also shown within the graph as opposed to the one in
Example~\ref{exp:parametric_regulatory_network}.

\begin{figure}[t]
  \centering

{ 

\begin{tikzpicture}[ampersand replacement=\&]
\matrix (m) [matrix of math nodes,row sep=0.5em,column sep=1.5em] {
   \&      \&    \&      \&    \& {} \&    \&      \&    \&      \&    \\
   \& \bf{0000} \&    \& 1000 \&    \&    \&    \& 0001 \&    \& 1001 \&    \\
   \&      \&    \&      \& {} \&    \& {} \&      \&    \&      \&    \\
   \& \bf{0010} \&    \& \bf{1010} \&    \&    \&    \& 0011 \&    \& 1011 \&    \\
   \&      \& {} \&      \&    \&    \&    \&      \& {} \&      \&    \\
{} \&      \&    \&      \&    \&    \&    \&      \&    \&      \& {} \\
   \&      \& {} \&      \&    \&    \&    \&      \& {} \&      \&    \\
   \& \bf{0100} \&    \& \bf{1100} \&    \&    \&    \& 0101 \&    \& 1101 \&    \\
   \&      \&    \&      \& {} \&    \& {} \&      \&    \&      \&    \\
   \& \bf{0110} \&    \& \bf{1110} \&    \&    \&    \& 0111 \&    \& 1111 \&    \\
};


\path[->,every node/.style={font=\scriptsize}]
(m-5-3) edge[bend right, thick] node[left] {$b+$} (m-7-3)
(m-7-3) edge[bend right, thick] node[right] {$b-$} (m-5-3)

(m-5-9) edge[bend right] node[left] {$b+$} (m-7-9)
(m-7-9) edge[bend right] node[right] {$b-$} (m-5-9)

(m-3-5) edge[bend right] node[below] {$d+$} (m-3-7)
(m-3-7) edge[bend right] node[above] {$d-$} (m-3-5)

(m-9-5) edge[bend right] node[below] {$d+$} (m-9-7)
(m-9-7) edge[bend right] node[above] {$d-$} (m-9-5)





(m-8-2) edge[thick] node [right] {$c+$} (m-10-2)
(m-8-4) edge node [right] {$c+$} (m-10-4)
(m-8-8) edge node [right] {$c+$} (m-10-8)
(m-8-10) edge node [right] {$c+$} (m-10-10)

(m-4-2) edge node [right] {$c-$} (m-2-2)
(m-4-4) edge node [right] {$c-$} (m-2-4)
(m-4-8) edge node [right] {$c-$} (m-2-8)
(m-4-10) edge node [right] {$c-$} (m-2-10)

(m-4-2) edge[thick] node[above] {$a+$} (m-4-4)
(m-8-2) edge[bend left=10, thick] node[above] {$a+, \parametrisation$} (m-8-4)
(m-10-2) edge[thick] node[above] {$a+$} (m-10-4)

(m-2-4) edge node[above] {$a-$} (m-2-2)
(m-2-10) edge node[above] {$a-$} (m-2-8)
(m-4-10) edge node[above] {$a-$} (m-4-8)
(m-8-4) edge[bend left=10] node[below] {$a-, {\parametrisation}'$} (m-8-2)
(m-8-10) edge node[above] {$a-$} (m-8-8)
(m-10-10) edge node[above] {$a-$} (m-10-8)
;
\end{tikzpicture}
}
  \caption{States and transitions of
  $(\dprn, \{\parametrisation,{\parametrisation}'\})$ depicted as
  nodes and edges of a state space graph respectively.
  Transitions changing the value of $b$ and $d$ are displayed schematically.
  Transitions only enabled by a single of the parametrisations are marked
  accordingly.
  Bold font and lines indicate states and transitions used by at least
  one minimal trace from the initial state to the goal.
  \label{fig:state_space_graph}}
\end{figure}
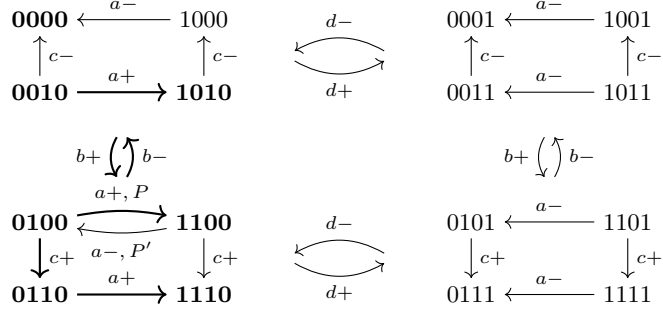

In our example, there are three minimal traces from the initial state $\state$
reaching the goal $a=1$:
\begin{align*}
  &\langle0000\rangle \xrightarrow{b+,\langle 0 \rangle}
    \langle0100\rangle \xrightarrow{a+,\langle 1 0 0 \rangle}
    \langle1100\rangle\\
  &\langle0000\rangle \xrightarrow{b+,\langle 0 \rangle}
    \langle0100\rangle \xrightarrow{c+,\langle 1 \rangle}
    \langle0110\rangle \xrightarrow{a+,\langle 1 1 0 \rangle}
    \langle1110\rangle\\
  &\langle0000\rangle \xrightarrow{b+,\langle 0 \rangle}
    \langle0100\rangle \xrightarrow{c+,\langle 1 \rangle}
    \langle0110\rangle \xrightarrow{b-,\langle 1 \rangle}
    \langle0010\rangle \xrightarrow{a+,\langle 0 1 0 \rangle}
    \langle1010\rangle
\end{align*}

All the listed traces share a common prefix, however, they are all minimal
as a different regulator state is used to activate $a$ each time, thus each of the
traces has at least one unique transition.
One may further remark that the first (shortest) minimal path is only available
under parametrisation $\parametrisation$, however,
thanks to Property~\ref{ppt:minimal_trace_independence}
this has no impact on the reduction procedure itself.

Observe that node $d$ never activates in any of the minimal traces.
This follows from the fact that $a$ is never allowed to activate while $d$ is active.
Thus, if $d$ activates it has to deactivate again before the goal can be reached.
As $d$ has no impact on the value of the other components besides $a$,
such an activation and deactivation loop can always be stripped from the trace
to obtain a smaller trace, unlike the loop by $b$ in the third (longest) trace,
which is necessary for the activation of $c$.
One might thus expect the activation of $d$ to be pruned during the reduction
procedure, which is, indeed the case:

We start with $\reducedObjectiveSet \assign \singletonSet{\explicitObjective{a}{0}{1}}$
according to rule (1) of Definition~\ref{def:reduction_procedure}.

Inference of the regulator cover set used for
$\validObjectiveTransitionSet{\state}{\explicitObjective{a}{0}{1}} =
\setDescription{\explicitPartialTransition{a}{0}{1}{\partialRegulatorState}}
{\partialRegulatorState \in
{\partialRegulatorStateSubset}_{\explicitTransitionValueChange{a}{0}{1}}} =
\{(\explicitTransitionValueChange{a}{0}{1}, \langle 100 \rangle),
(\explicitTransitionValueChange{a}{0}{1}, \langle 010 \rangle),
(\explicitTransitionValueChange{a}{0}{1}, \langle 110 \rangle)\}$
is illustrated in Example~\ref{exp:regulation_cover_set}.
Then, by rule (2) of Definition~\ref{def:reduction_procedure}, the following
objectives are included in $\reducedObjectiveSet \assign \reducedObjectiveSet
\cup \{\explicitObjective{b}{0}{0}, \explicitObjective{b}{0}{1},
\explicitObjective{c}{0}{0}, \explicitObjective{c}{0}{1},
\explicitObjective{d}{0}{0}\}$.

For arbitrary component $v$, the objective $\explicitObjective{v}{0}{0}$
has an empty valid transition set $\validObjectiveTransitionSet{\state}
{\explicitObjective{v}{0}{0}}=\emptyset$ and thus neither of rules (2) or (3)
are applicable.
For the the remaining $\explicitObjective{b}{0}{1}$ and
$\explicitObjective{c}{0}{1}$ rule (2) produces only duplicate objectives
($\explicitObjective{b}{0}{0}$ and $\explicitObjective{b}{0}{1}$, respectively).
Rule (3), however, may be applied to $\explicitObjective{b}{0}{1}$ and
$\explicitObjective{c}{0}{1}$ to bridge them to $\explicitObjective{b}{0}{0}$
and $\explicitObjective{c}{0}{0}$, respectively, to include objectives
$\reducedObjectiveSet \assign \reducedObjectiveSet \cup
\{\explicitObjective{b}{1}{0}, \explicitObjective{c}{1}{0}\}$.

Only duplicate objectives are obtained by application of either rule (2) or (3) on the newly added $\explicitObjective{b}{1}{0}$ and $\explicitObjective{c}{1}{0}$.
Thus, the reduction concludes with
$\reducedObjectiveSet = \{\explicitObjective{a}{0}{1},
\explicitObjective{b}{0}{0}, \explicitObjective{b}{0}{1},
\explicitObjective{b}{1}{0}, \explicitObjective{c}{0}{0},
\explicitObjective{c}{0}{1}, \explicitObjective{c}{1}{0},
\explicitObjective{d}{0}{0}\}$,
with valid transition set
$\validObjectiveTransitionSet{\state}{\reducedObjectiveSet} =
\{(\explicitTransitionValueChange{a}{0}{1}, \langle 100 \rangle),
(\explicitTransitionValueChange{a}{0}{1}, \langle 010 \rangle),
(\explicitTransitionValueChange{a}{0}{1}, \langle 110 \rangle),
(\explicitTransitionValueChange{b}{0}{1}, \langle 0 \rangle),
(\explicitTransitionValueChange{b}{1}{0}, \langle 1 \rangle),
(\explicitTransitionValueChange{c}{0}{1}, \langle 1 \rangle),
(\explicitTransitionValueChange{c}{1}{0}, \langle 0 \rangle)\}$.
One may observe that the computed transition set indeed covers all
the transitions used by any of the minimal traces
(thick edges in Figure~\ref{fig:state_space_graph}).

Finally, the limit vectors for the new DPRN ${\dprn}' =
(\prn, {\activationLimitVector}', {\inhibitionLimitVector}')$
are computed as follows:
\begin{align*}
  {\activationLimitVector}' &= \langle a=1, b=1, c=1, d=\minusInfinity \rangle\\
  {\inhibitionLimitVector}' &= \langle a=\infinity, b=0, c=0, d=\infinity \rangle
\end{align*}

Observe that component $d$ is indeed completely forbidden from acting in the
reduced model, considerably decreasing the reachable state space that
has to be explored.
Notice that deactivation of $a$ is also disabled, however,
in our Boolean case this has no practical effect w.r.t. reachability of $a=1$.

\end{example}
}


\subsection{Correctness}
\label{sec:reduction_correctness}

Following the interpretation of the reduction procedure and thanks to the
monotonicity of value updating, a transition $\transitionTuple$
remains enabled in ${\dprn}'$ iff at least one partial transition
$\partialTransitionTuple$ exists in
$\validObjectiveTransitionSet{\state}{\reducedObjectiveSet}$
with $\regulatorState \in \partialRegulatorState$.
This leads us to formulate the soundness theorem of the reduction procedure,
guaranteeing that all transitions of all minimal traces are preserved
and thus, in turn, all minimal traces are preserved.


{ 

\begin{theorem}
\label{thm:minimal_trace_preservation}

Let $\dprn$ be a DPRN,
and let a realisable trace $\trace$ of $\dprn$ be minimal for an initial state
$\state \in \stateSet{\dprn}$ and goal $\goal$.
Then, for any transition $\transitionTuple \in \collapse{\trace}$
there exists at least one partial transition
$(\transitionValueChange, \partialRegulatorState) \in
\objectiveTransitionSet{\reducedObjectiveSet}$
such that $\regulatorState \in \partialRegulatorState$,
where $\reducedObjectiveSet$
is constructed according to Definition~\ref{def:reduction_procedure}.

\end{theorem}

}

The proof of the theorem relies on showing that any transition
which is not preserved is part of a cycle on any trace leading to the goal,
and as consequence does not belong to any minimal trace.
Due to the complexity, the formal proof is given in
Appendix~\ref{apx:minimal_trace_preservation_proof}.


\section{Regulation Cover Set Inference}
\label{sec:regulation_cover_set_inference}

In this section we introduce a heuristic for construction of regulation cover sets whose size,
w.r.t. specified regulator values across all partial regulator states,
does not exceed the size of the concrete regulation cover set.

Let ${\partialRegulatorStateSubset}_{\textsc{ena}} =
\setDescription{\partialRegulatorState \in \partialRegulatorStateSet{\component}}
{\forall \regulatorState \in \partialRegulatorState:
\transitionTuple \ \text{is enabled}}$ be the set of all
partial regulators states which contain no bad regulator states.
For each $\index \in \valueSet{0}{\size{\regulatorSet{\component}}}$
let ${\partialRegulatorStateSubset}_{\index} = \setDescription
{\partialRegulatorState \in \partialRegulatorStateSet{\component}}
{\size{\setDescription{\regulator \in \regulatorSet{\component}}
{\vectorElement{\partialRegulatorState}{\regulator} = \anyValue}} = \index}$
to be the set of all partial regulator states with exactly $\index$
regulator values equal to $\anyValue$.

The algorithm consists of choosing partial regulator state set,
${\partialRegulatorStateSubset}_{\textsc{ext}}$, to cover each
(concrete) regulator state enabling the value change.
This is done separately for each regulator state in an increasing
order of a weight function.
The weight function represents flexibility of covering the regulator state,
i.e. there are more partial regulator states in
${\partialRegulatorStateSubset}_{\textsc{ena}}$ containing a regulator state
with a larger weight than the ones covering regulator state with smaller weight.
The weights are dynamic as the partial regulator states get removed
(${\partialRegulatorStateSubset}_{\textsc{rmv}}$) throughout the algorithm.
The ${\partialRegulatorStateSubset}_{\textsc{ext}}$ for each regulator state
is computed by testing candidate sets of partial regulator states
from ${\partialRegulatorStateSubset}_{\index}$ in decreasing order on $\index$.
A cover set for each regulator state is guaranteed to exist as for $\index = 0$
the candidate set is a singleton set containing the regulator state itself.
Once a suitable cover set is found for a particular regulator state,
it is included in the regulation cover set
${\partialRegulatorStateSubset}_{\transitionValueChange}$
and all partial regulator states containing
the regulator state are excluded from further computation.

As the weight function gives only a partial order on the regulator states,
the algorithm is forced to make nondeterministic choices.
This occurs, however, only in cases when the choices are isomorphic.
As such, the partial order given by weights can be extended to a total order
arbitrarily, e.g. by underlying lexicographic order.
The pseudocode of the algorithm to construct regulation cover sets
is given in Algorithm \ref{alg:cover_set_computation}.


\begin{algorithm}[t]
  \begin{algorithmic}
    \Function{Weight}{$\regulatorState$}
      \State \Return $\size{\setDescription{\partialRegulatorState \in
      ({\partialRegulatorStateSubset}_{1} \cap
      {\partialRegulatorStateSubset}_{\textsc{ena}}) \setminus
      {\partialRegulatorStateSubset}_{\textsc{rmv}}}
      {\regulatorState \in \partialRegulatorState}} +
      \frac{\size{\setDescription{\partialRegulatorState \in
      {\partialRegulatorStateSubset}_{1} \cap
      {\partialRegulatorStateSubset}_{\textsc{ena}}}
      {\regulatorState \in \partialRegulatorState}}}
      {\size{\regulatorSet{\component}} + 1}$
    \EndFunction
    \State
    \Function{ComputeCoverSet}
    {$\transitionValueChange =
    \explicitTransitionValueChange{\component}{\myValue}{\secondValue}$}
      \State ${\partialRegulatorStateSubset}_{\transitionValueChange}
      \gets \emptyset$
      \State ${\partialRegulatorStateSubset}_{\textsc{rmv}}
      \gets \emptyset$

      \While{${\partialRegulatorStateSubset}_{0} \neq \emptyset$}
        \State $\regulatorState \gets
        {\regulatorState}' \in ({\partialRegulatorStateSubset}_{0} \cap
        {\partialRegulatorStateSubset}_{\textsc{ena}}) \setminus
        {\partialRegulatorStateSubset}_{\textsc{rmv}}$
        \State\hspace{5mm}
       with \Call{Weight}{${\regulatorState}'$}
        $= \textsc{min}\{$\Call{Weight}{${\regulatorState}''$}
        $ \mid {\regulatorState}'' \in ({\partialRegulatorStateSubset}_{0} \cap
        {\partialRegulatorStateSubset}_{\textsc{ena}}) \setminus
        {\partialRegulatorStateSubset}_{\textsc{rmv}} \}$
        \State ${\partialRegulatorStateSubset}_{\textsc{ext}} \gets \emptyset$
        \State $\index \gets \size{\regulatorSet{\component}} - 1$
        \While{$\omega$ is not covered by
        ${\partialRegulatorStateSubset}_{\transitionValueChange} \cup
        {\partialRegulatorStateSubset}_{\textsc{ext}}$}
            \State ${\partialRegulatorStateSubset}_{\textsc{ext}} \gets
              ({\partialRegulatorStateSubset}_{\index} \cap
              {\partialRegulatorStateSubset}_{\textsc{ena}}) \setminus
              {\partialRegulatorStateSubset}_{\textsc{rmv}}$
            \State $\index \gets \index - 1$
        \EndWhile
        \State ${\partialRegulatorStateSubset}_{\transitionValueChange} \gets
        {\partialRegulatorStateSubset}_{\transitionValueChange} \cup
        {\partialRegulatorStateSubset}_{\textsc{ext}}$
        \State ${\partialRegulatorStateSubset}_{\textsc{rmv}} \gets
        {\partialRegulatorStateSubset}_{\textsc{rmv}} \cup
        \setDescription{\partialRegulatorState \in
        \partialRegulatorStateSet{\component}}
        {\regulatorState \in \partialRegulatorState}$
      \EndWhile

      \State \Return ${\partialRegulatorStateSubset}_{\transitionValueChange}$
    \EndFunction
  \end{algorithmic}
  \caption{Pseudocode of the algorithm computing regulation cover set.
  \label{alg:cover_set_computation}
  }
\end{algorithm}

The correctness of the algorithm comes directly from the construction.
No bad states may be included as the algorithm works only with the set of
partial regulator states which include no bad states.
On the other hand, all regulator states which enable the value change are
fully covered as the algorithm ensures this for each of them individually.

The resulting cover set computed by Algorithm \ref{alg:cover_set_computation}
contains no more explicit regulator value specifications than the concrete
regulation cover set.
This is a consequence of the order of regulator states covering.
Suppose a regulator state $\regulatorState$ is covered by
several partial regulator states which contain more regulator value
specifications than $\regulatorState$ itself.
Each partial regulator state
$\partialRegulatorState \in {\partialRegulatorStateSubset}_{1}$
with $\vectorElement{\partialRegulatorState}{\regulator} = \anyValue$
is shared with exactly $\maximum{\regulator} - 1$ other regulator states.
Thus, the partial regulator states included to cover $\regulatorState$
can be utilised while covering $\maximum{\regulator} - 1$ other regulator states.
Finally, since $\textsc{Weight}(\regulatorState) \geq 2$ is the smallest
weight among all uncovered regulator states,
all the other uncovered regulator states are also sharing partial regulator
states among themselves, thus closing the loop and guaranteeing the
regulator value specification debt eventually gets "payed off".

The fractional part of the weight function is included to introduce bias
towards states that have less partial regulator states in the beginning
(due to sharing with more bad states).
If there are two regulator states $\regulatorState$ and ${\regulatorState}'$ such that
$\lfloor \textsc{Weight}(\regulatorState) \rfloor =
  \lfloor \textsc{Weight}({\regulatorState}') \rfloor$ but
$\textsc{Weight}(\regulatorState) < \textsc{Weight}({\regulatorState}')$,
we know that both of them have equally many partial regulator states to choose
from for their respective cover sets.
However, more of the partial regulator states containing ${\regulatorState}'$
have been removed and thus, quite possibly included in the regulation cover set
${\partialRegulatorStateSubset}_{\transitionValueChange}$.
${\regulatorState}'$ is therefore in all likelihood already covered
to a higher degree than $\regulatorState$ and possibly, has more covering options.
The bias thus ensures $\regulatorState$ is covered first in order to avoid
introducing potentially redundant partial regulator states into the regulation
cover set.

Both principles making up the weight function are illustrated in
Example~\ref{exp:regulation_cover_set}.


{ 
\def\someRegState{\orderedCouple{\component}{\regulatorState}}
\def\orderedTuple#1{\langle #1 \rangle}

\begin{example}
\label{exp:regulation_cover_set}

Consider the same directed parametric regulatory network $\dprn$ as in
Example~\ref{exp:reduction_procedure}.

We now show the regulation cover set computation for value changes of component $a$.
Let us start with $\explicitTransitionValueChange{a}{0}{1}$.
The initial configuration and first two iterations, consisting of
covering of the first two regulator states, of the algorithm are
schematically depicted in Figure~\ref{fig:activation_regulation_cover_set}.

\begin{figure}[t]
  \centering
  \begin{subfigure}{.3\linewidth}

{ 

\begin{tikzpicture}[ampersand replacement=\&]
\tikzset{fore/.style={preaction={draw=white, -,line width=4pt}}}
\matrix (m) [matrix of math nodes,row sep=2em,column sep=2em] {
    \& \bf{110} \&          \\
111 \& \bf{100} \& \bf{010} \\
101 \& 011      \& 000      \\
    \& 001      \&          \\
};

\path[]
(m-1-2) edge[fore, line width=2pt] (m-2-2)
(m-2-2) edge[fore] (m-3-1)
(m-2-2) edge[fore] (m-3-3)

(m-1-2) edge[fore] node[above left=-0.1cm] {$1 1 \anyValue$} (m-2-1)
(m-1-2) edge[fore, line width=2pt] (m-2-3)
(m-2-1) edge[fore] node[left=-0.1cm] {$1 \anyValue 1$} (m-3-1)
(m-2-3) edge[fore] (m-3-3)
(m-3-1) edge[fore] node[below left=-0.1cm] {$\anyValue 0 1$} (m-4-2)
(m-3-3) edge[fore] (m-4-2)

(m-2-1) edge[fore] (m-3-2)
(m-2-3) edge[fore] (m-3-2)
(m-3-2) edge[fore] (m-4-2)
;

\end{tikzpicture}
}
    \caption{Initial configuration}
  \end{subfigure}
  \begin{subfigure}{.3\linewidth}

{ 

\begin{tikzpicture}[ampersand replacement=\&]
\tikzset{fore/.style={preaction={draw=white, -,line width=4pt}}}
\matrix (m) [matrix of math nodes,row sep=2em,column sep=2em] {
    \& \bf{110} \&          \\
111 \& \bf{100} \& \underline{\bf{010}} \\
101 \& 011      \& 000      \\
    \& 001      \&          \\
};

\path[]
(m-1-2) edge[fore, line width=2pt] (m-2-2)
(m-2-2) edge[fore] (m-3-1)
(m-2-2) edge[fore] (m-3-3)

(m-1-2) edge[fore] (m-2-1)
(m-1-2) edge[fore, line width=2pt, dashed] (m-2-3)
(m-2-1) edge[fore] (m-3-1)
(m-2-3) edge[fore, dashed] (m-3-3)
(m-3-1) edge[fore] (m-4-2)
(m-3-3) edge[fore] (m-4-2)

(m-2-1) edge[fore] (m-3-2)
(m-2-3) edge[fore, dashed] (m-3-2)
(m-3-2) edge[fore] (m-4-2)
;

\end{tikzpicture}
}
    \caption{Configuration after one iteration.}
  \end{subfigure}
  \begin{subfigure}{.3\linewidth}

{ 

\begin{tikzpicture}[ampersand replacement=\&]
\tikzset{fore/.style={preaction={draw=white, -,line width=4pt}}}
\matrix (m) [matrix of math nodes,row sep=2em,column sep=2em] {
    \& \bf{110} \&          \\
111 \& \underline{\bf{100}} \& \bf{010} \\
101 \& 011      \& 000      \\
    \& 001      \&          \\
};

\path[]
(m-1-2) edge[fore, line width=2pt, dashed] (m-2-2)
(m-2-2) edge[fore, dashed] (m-3-1)
(m-2-2) edge[fore, dashed] (m-3-3)

(m-1-2) edge[fore] (m-2-1)
(m-1-2) edge[fore, line width=2pt, dashed] (m-2-3)
(m-2-1) edge[fore] (m-3-1)
(m-2-3) edge[fore, dashed] (m-3-3)
(m-3-1) edge[fore] (m-4-2)
(m-3-3) edge[fore] (m-4-2)

(m-2-1) edge[fore] (m-3-2)
(m-2-3) edge[fore, dashed] (m-3-2)
(m-3-2) edge[fore] (m-4-2)
;

\end{tikzpicture}
}
    \caption{Configuration after two iterations.}
  \end{subfigure}
  \caption{Regulator states of component $a$ during computation of regulation
  cover set for value change $\explicitTransitionValueChange{a}{0}{1}$.
  Only the leftmost edges in (a) are labelled by the corresponding partial regulator states
  $1 1 \anyValue, 1 \anyValue 1$ and $\anyValue 0 1$ for the sake of readability.
  Bold text and lines indicate (partial) regulator states which enable the
  value change (${\partialRegulatorStateSubset}_{\textsc{ena}}$).
  Underlined regulator state is the state covered in the respective iteration
  and dashed lines represent removed partial regulator states
  (${\partialRegulatorStateSubset}_{\textsc{rmv}}$).
  \label{fig:activation_regulation_cover_set}}
\end{figure}
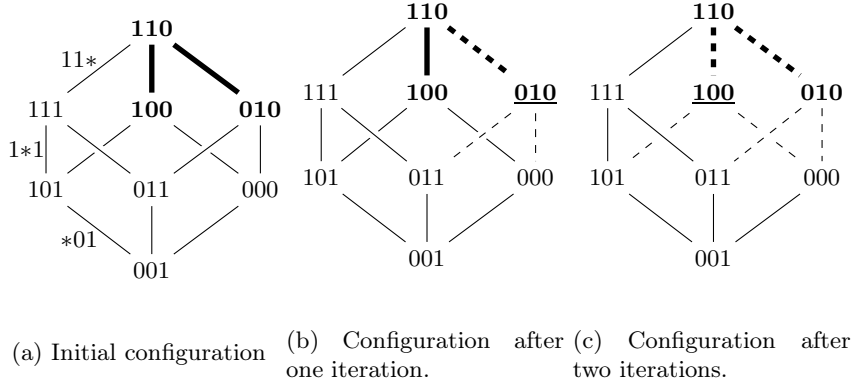

Figure~\ref{fig:activation_regulation_cover_set}
lists all regulator states of component $a$ as nodes in a graph.
Bold font indicates the three regulator states which enable the increase of $a$.
The partial regulator states from ${\partialRegulatorStateSubset}_{1}$
correspond to edges in the graph, connecting contained regulator states.
Thick edges indicate partial regulator states which contain no bad
regulator states.
Partial regulator states from ${\partialRegulatorStateSubset}_{2}$ could in turn
be viewed as squares in the diagram, all of them containing at least one bad
regulator state in our case.
In the graphical representation of regulator states,
a partial regulator state belonging to ${\partialRegulatorStateSubset}_{\index}$
is a $\index$-dimensional hypercube in the Boolean case,
or a $\index$-dimensional hyper-rectangular cuboid in the general case.

The graph representation in Figure~\ref{fig:activation_regulation_cover_set}
allows for easy visualisation of the weight function.
The weight corresponds to number of thick, non-dashed edges
plus, the number of thick edges divided by $\size{\regulatorSet{a}} + 1$,
in our case $4$.
Consequently, in the initial configuration
(Figure~\ref{fig:activation_regulation_cover_set}~(a))
the regulator states $\orderedTuple{1 0 0}$ and $\orderedTuple{0 1 0}$
have equal (minimal) weight.
This is justified by their perfectly symmetrical position.

Figure~\ref{fig:activation_regulation_cover_set} illustrates the run of the
algorithm assuming lexicographic order is used to distinguish between
regulator states with equal weights.
In the first iteration $\orderedTuple{0 1 0}$ is covered using itself for
the extension set ${\partialRegulatorStateSubset}_{\textsc{ext}} =
\singletonSet{\orderedTuple{0 1 0}}$ as the only partial regulator state
with more unspecified regulator values, $\orderedTuple{\anyValue 1 0}$,
alone does not fully cover $\orderedTuple{0 1 0}$.
Figure~\ref{fig:activation_regulation_cover_set}~(b) depicts the situation
after the first iteration, including the removed partial regulator states
(dashed lines).

In the second iteration $\orderedTuple{1 0 0}$ is covered in the exact same fashion,
owning to the symmetric position w.r.t. $\orderedTuple{0 1 0}$.
The result is shown in Figure~\ref{fig:activation_regulation_cover_set}~(c).

No partial regulator states remain for the last regulator state
$\orderedTuple{1 1 0}$ except the regulator state itself.
Thus, $\orderedTuple{1 1 0}$ also gets covered explicitly.
The algorithm therefore concludes with the concrete regulation cover set
${\partialRegulatorStateSubset}_{\explicitTransitionValueChange{a}{0}{1}} =
\{\orderedTuple{0 1 0}, \orderedTuple{1 0 0}, \orderedTuple{1 1 0}\}$,
which, in fact, is the optimal solution in our case.

Let us now consider also the decreasing case
$\explicitTransitionValueChange{a}{1}{0}$.
Again, we illustrate the running of the algorithm using graph representation
of the regulator states of $a$.
All iterations up to the final one of the algorithm using lexicographic order
on regulator states of equal weight are given in
Figure~\ref{fig:inhibition_regulation_cover_set}.

\begin{figure}[t]
  \centering
  \begin{subfigure}{.3\linewidth}

{ 

\begin{tikzpicture}[ampersand replacement=\&]
\tikzset{fore/.style={preaction={draw=white, -,line width=4pt}}}
\matrix (m) [matrix of math nodes,row sep=2em,column sep=2em] {
         \& 110      \&          \\
\bf{111} \& \bf{100} \& 010      \\
\bf{101} \& \bf{011} \& \bf{000} \\
         \& \bf{001} \&          \\
};

\path[fill=gray!50,opacity=.35] (m-2-2.south) to (m-3-3.west) to (m-4-2.north)
  to (m-3-1.east) to (m-2-2.south);

\path[]
(m-1-2) edge[fore] (m-2-2)
(m-2-2) edge[fore, line width=2pt] (m-3-1)
(m-2-2) edge[fore, line width=2pt] (m-3-3);

\path[fill=gray!50,opacity=.35] (m-2-1.south east) to (m-3-2.north west)
to (m-4-2.north west) to (m-3-1.south east) to (m-2-1.south east);

\path[]
(m-1-2) edge[fore] (m-2-1)
(m-1-2) edge[fore] (m-2-3)
(m-2-1) edge[fore, line width=2pt] (m-3-1)
(m-2-3) edge[fore] (m-3-3)
(m-3-1) edge[fore, line width=2pt] (m-4-2)
(m-3-3) edge[fore, line width=2pt] (m-4-2)

(m-2-1) edge[fore, line width=2pt] (m-3-2)
(m-2-3) edge[fore] (m-3-2)
(m-3-2) edge[fore, line width=2pt] (m-4-2)
;

\end{tikzpicture}
}
    \caption{Initial configuration}
  \end{subfigure}
  \begin{subfigure}{.3\linewidth}

{ 

\begin{tikzpicture}[ampersand replacement=\&]
\tikzset{fore/.style={preaction={draw=white, -,line width=4pt}}}
\matrix (m) [matrix of math nodes,row sep=2em,column sep=2em] {
         \& 110      \&          \\
\bf{111} \& \bf{100} \& 010      \\
\bf{101} \& \bf{011} \& \underline{\bf{000}} \\
         \& \bf{001} \&          \\
};

\path[fill=gray!50, opacity=.35, pattern=north east lines] (m-2-2.south) to (m-3-3.west)
  to (m-4-2.north) to (m-3-1.east) to (m-2-2.south);

\path[]
(m-1-2) edge[fore] (m-2-2)
(m-2-2) edge[fore, line width=2pt] (m-3-1)
(m-2-2) edge[fore, line width=2pt, dashed, double] (m-3-3);

\path[fill=gray!50,opacity=.35] (m-2-1.south east) to (m-3-2.north west)
to (m-4-2.north west) to (m-3-1.south east) to (m-2-1.south east);

\path[]

(m-1-2) edge[fore] (m-2-1)
(m-1-2) edge[fore] (m-2-3)
(m-2-1) edge[fore, line width=2pt] (m-3-1)
(m-2-3) edge[fore, dashed] (m-3-3)
(m-3-1) edge[fore, line width=2pt] (m-4-2)
(m-3-3) edge[fore, line width=2pt, dashed, double] (m-4-2)

(m-2-1) edge[fore, line width=2pt] (m-3-2)
(m-2-3) edge[fore] (m-3-2)
(m-3-2) edge[fore, line width=2pt] (m-4-2)
;

\end{tikzpicture}
}
    \caption{Configuration after one iteration.}
  \end{subfigure}
  \begin{subfigure}{.3\linewidth}

{ 

\begin{tikzpicture}[ampersand replacement=\&]
\tikzset{fore/.style={preaction={draw=white, -,line width=4pt}}}
\matrix (m) [matrix of math nodes,row sep=2em,column sep=2em] {
         \& 110      \&          \\
\bf{111} \& \underline{\bf{100}} \& 010      \\
\bf{101} \& \bf{011} \& \bf{000} \\
         \& \bf{001} \&          \\
};

\path[fill=gray!50, opacity=.35, pattern=north east lines] (m-2-2.south) to (m-3-3.west)
  to (m-4-2.north) to (m-3-1.east) to (m-2-2.south);

\path[]
(m-1-2) edge[fore, dashed] (m-2-2)
(m-2-2) edge[fore, line width=2pt, dashed, double] (m-3-1)
(m-2-2) edge[fore, line width=2pt, dashed, double] (m-3-3);

\path[fill=gray!50,opacity=.35] (m-2-1.south east) to (m-3-2.north west)
to (m-4-2.north west) to (m-3-1.south east) to (m-2-1.south east);

\path[]

(m-1-2) edge[fore] (m-2-1)
(m-1-2) edge[fore] (m-2-3)
(m-2-1) edge[fore, line width=2pt] (m-3-1)
(m-2-3) edge[fore, dashed] (m-3-3)
(m-3-1) edge[fore, line width=2pt] (m-4-2)
(m-3-3) edge[fore, line width=2pt, dashed, double] (m-4-2)

(m-2-1) edge[fore, line width=2pt] (m-3-2)
(m-2-3) edge[fore] (m-3-2)
(m-3-2) edge[fore, line width=2pt] (m-4-2)
;

\end{tikzpicture}
}
    \caption{Configuration after two iterations.}
  \end{subfigure}
  \begin{subfigure}{.3\linewidth}

{ 

\begin{tikzpicture}[ampersand replacement=\&]
\tikzset{fore/.style={preaction={draw=white, -,line width=4pt}}}
\tikzset{foredouble/.style={preaction={draw=white, -,line width=3pt, double}}}
\matrix (m) [matrix of math nodes,row sep=2em,column sep=2em] {
         \& 110      \&          \\
\bf{111} \& \bf{100} \& 010      \\
\bf{101} \& \underline{\bf{011}} \& \bf{000} \\
         \& \bf{001} \&          \\
};

\path[fill=gray!50, opacity=.35, pattern=north east lines] (m-2-2.south)
  to (m-3-3.west) to (m-4-2.north) to (m-3-1.east) to (m-2-2.south);

\path[]
(m-1-2) edge[fore, dashed] (m-2-2)
(m-2-2) edge[foredouble, line width=2pt, dashed, double] (m-3-1)
(m-2-2) edge[foredouble, line width=2pt, dashed, double] (m-3-3);

\path[fill=gray!50, opacity=.35, pattern=vertical lines] (m-2-1.south east)
 to (m-3-2.north west) to (m-4-2.north west) to (m-3-1.south east)
 to (m-2-1.south east);

\path[]
(m-1-2) edge[fore] (m-2-1)
(m-1-2) edge[fore] (m-2-3)
(m-2-1) edge[fore, line width=2pt] (m-3-1)
(m-2-3) edge[fore, dashed] (m-3-3)
(m-3-1) edge[fore, line width=2pt] (m-4-2)
(m-3-3) edge[foredouble, line width=2pt, dashed, double] (m-4-2)

(m-2-1) edge[foredouble, line width=2pt, dashed, double] (m-3-2)
(m-2-3) edge[fore, dashed] (m-3-2)
(m-3-2) edge[foredouble, line width=2pt, dashed, double] (m-4-2)
;

\end{tikzpicture}
}
    \caption{Configuration after three iterations.}
  \end{subfigure}
  \begin{subfigure}{.3\linewidth}

{ 

\begin{tikzpicture}[ampersand replacement=\&]
\tikzset{fore/.style={preaction={draw=white, -,line width=4pt}}}
\tikzset{foredouble/.style={preaction={draw=white, -,line width=3pt, double}}}
\matrix (m) [matrix of math nodes,row sep=2em,column sep=2em] {
         \& 110      \&          \\
\underline{\bf{111}} \& \bf{100} \& 010      \\
\bf{101} \& \bf{011} \& \bf{000} \\
         \& \bf{001} \&          \\
};

\path[fill=gray!50, opacity=.35, pattern=north east lines] (m-2-2.south)
  to (m-3-3.west) to (m-4-2.north) to (m-3-1.east) to (m-2-2.south);

\path[]
(m-1-2) edge[fore, dashed] (m-2-2)
(m-2-2) edge[foredouble, line width=2pt, dashed, double] (m-3-1)
(m-2-2) edge[foredouble, line width=2pt, dashed, double] (m-3-3);

\path[fill=gray!50, opacity=.35, pattern=vertical lines] (m-2-1.south east)
 to (m-3-2.north west) to (m-4-2.north west) to (m-3-1.south east)
 to (m-2-1.south east);

\path[]
(m-1-2) edge[fore, dashed] (m-2-1)
(m-1-2) edge[fore] (m-2-3)
(m-2-1) edge[foredouble, line width=2pt, dashed, double] (m-3-1)
(m-2-3) edge[fore, dashed] (m-3-3)
(m-3-1) edge[fore, line width=2pt] (m-4-2)
(m-3-3) edge[foredouble, line width=2pt, dashed, double] (m-4-2)

(m-2-1) edge[foredouble, line width=2pt, dashed, double] (m-3-2)
(m-2-3) edge[fore, dashed] (m-3-2)
(m-3-2) edge[foredouble, line width=2pt, dashed, double] (m-4-2)
;

\end{tikzpicture}
}
    \caption{Configuration after four iterations.}
  \end{subfigure}
  \begin{subfigure}{.3\linewidth}

{ 

\begin{tikzpicture}[ampersand replacement=\&]
\tikzset{fore/.style={preaction={draw=white, -,line width=4pt}}}
\tikzset{foredouble/.style={preaction={draw=white, -,line width=3pt, double}}}
\matrix (m) [matrix of math nodes,row sep=2em,column sep=2em] {
         \& 110      \&          \\
\bf{111} \& \bf{100} \& 010      \\
\bf{101} \& \bf{011} \& \bf{000} \\
         \& \underline{\bf{001}} \&          \\
};

\path[fill=gray!50, opacity=.35, pattern=north east lines] (m-2-2.south)
  to (m-3-3.west) to (m-4-2.north) to (m-3-1.east) to (m-2-2.south);

\path[]
(m-1-2) edge[fore, dashed] (m-2-2)
(m-2-2) edge[foredouble, line width=2pt, dashed, double] (m-3-1)
(m-2-2) edge[foredouble, line width=2pt, dashed, double] (m-3-3);

\path[fill=gray!50, opacity=.35, pattern=vertical lines] (m-2-1.south east)
 to (m-3-2.north west) to (m-4-2.north west) to (m-3-1.south east)
 to (m-2-1.south east);

\path[]
(m-1-2) edge[fore, dashed] (m-2-1)
(m-1-2) edge[fore] (m-2-3)
(m-2-1) edge[foredouble, line width=2pt, dashed, double] (m-3-1)
(m-2-3) edge[fore, dashed] (m-3-3)
(m-3-1) edge[fore, line width=2pt, dashed] (m-4-2)
(m-3-3) edge[foredouble, line width=2pt, dashed, double] (m-4-2)

(m-2-1) edge[foredouble, line width=2pt, dashed, double] (m-3-2)
(m-2-3) edge[fore, dashed] (m-3-2)
(m-3-2) edge[foredouble, line width=2pt, dashed, double] (m-4-2)
;

\end{tikzpicture}
}
    \caption{Configuration after five iterations.}
  \end{subfigure}
  \caption{Regulator states of component $a$ during computation of regulation
  cover set for value change $\explicitTransitionValueChange{a}{1}{0}$.
  Bold text, lines and shaded areas indicate (partial) regulator states which
  enable the value change (${\partialRegulatorStateSubset}_{\textsc{ena}}$).
  The underlined regulator state is the state covered in the respective iteration.
  Dashes represent removed partial regulator states
  (${\partialRegulatorStateSubset}_{\textsc{rmv}}$) and double lines
  represent partial regulator states included in the regulation cover set
  (${\partialRegulatorStateSubset}_{\explicitTransitionValueChange{a}{1}{0}}$).
  \label{fig:inhibition_regulation_cover_set}}
\end{figure}
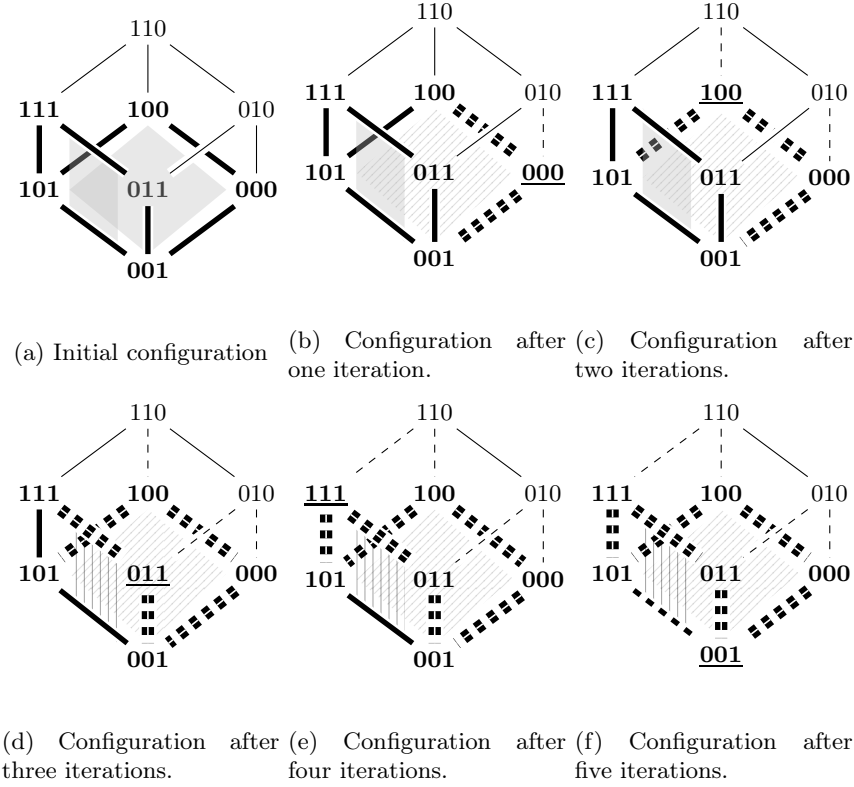

The algorithm begins with covering the regulator state $\orderedTuple{0 0 0}$.
Unlike in the case of increasing $a$, a nonempty candidate extension set exists
for partial regulator states on level ${\partialRegulatorStateSubset}_{2}$
containing a single element $\singletonSet{\orderedTuple{\anyValue 0 \anyValue}}$.
This partial regulator state alone, however, does not suffice to cover
$\orderedTuple{0 0 0}$ and extension set
$\{\orderedTuple{0 0 \anyValue}, \orderedTuple{\anyValue 0 0}\}$
is used instead as indicated by double lines in
Figure~\ref{fig:inhibition_regulation_cover_set} (b).
Notice that in this case, the node $\orderedTuple{0 0 0}$ gets covered by
two partial regulator states having one more regulator value specification
(a total of $4$ specifications against the explicit $3$).

According to the weight function, $\orderedTuple{1 0 0}$ gets covered next.
$\orderedTuple{\anyValue 0 \anyValue}$ is no longer available,
thus the first nonempty candidate extension set is
$\singletonSet{\orderedTuple{1 0 \anyValue}}$.
Although $\orderedTuple{1 0 \anyValue}$ alone is not enough to fully cover
$\orderedTuple{1 0 0}$, the cover set
${\partialRegulatorStateSubset}_{\explicitTransitionValueChange{a}{1}{0}}$
already contains $\orderedTuple{\anyValue 0 0}$ which covers
$\orderedTuple{1 0 0}$ completely in combination with
$\orderedTuple{1 0 \anyValue}$.
Thus, $\orderedTuple{1 0 0}$ gets covered by including only $2$ additional
regulator value specifications, effectively "paying-off" the depth incurred
while covering $\orderedTuple{0 0 0}$.

Covering $\orderedTuple{0 1 1}$ and subsequently $\orderedTuple{1 1 1}$
is identical to that of $\orderedTuple{0 0 0}$ and $\orderedTuple{1 0 0}$.
Both of them thus get covered by three partial regulator states
$\orderedTuple{0 \anyValue 1}$, $\orderedTuple{\anyValue 1 1}$ and
$\orderedTuple{1 \anyValue 1}$ as shown in
Figure~\ref{fig:inhibition_regulation_cover_set} (d) and (e).
Furthermore, $\orderedTuple{0 0 \anyValue}$ and $\orderedTuple{0 \anyValue 1}$,
fully cover $\orderedTuple{0 0 1}$ and $\orderedTuple{1 0 \anyValue}$,
$\orderedTuple{1 \anyValue 1}$ fully covers $\orderedTuple{1 0 1}$.
As such, the remaining two regulator states are covered with empty extension
sets and the final solution uses $12$ regulator value specifications
as opposed to the $18$ required by the explicit representation.

The fractional part of the weight function is crucial to distinguish between
$\orderedTuple{0 0 1}$, $\orderedTuple{1 0 1}$ and
$\orderedTuple{0 1 1}$, $\orderedTuple{1 1 1}$ after the second iteration
(Figure~\ref{fig:inhibition_regulation_cover_set} (c)).
Covering $\orderedTuple{0 0 1}$ or $\orderedTuple{1 0 1}$ before
$\orderedTuple{0 1 1}$ and $\orderedTuple{1 1 1}$ would include
either $\orderedTuple{\anyValue \anyValue 1}$ or $\orderedTuple{\anyValue 0 1}$,
depending on the exact order, in the final regulation cover set.
As both of them are redundant, this would lead to a suboptimal solution.

\end{example}
}

Algorithm \ref{alg:cover_set_computation} is quasilinear in the number of
regulator states and quadratic in the number of regulators.
Its main complexity comes from computing the
extension sets ${\partialRegulatorStateSubset}_{\textsc{ext}}$.
Whether a regulator state $\regulatorState \in \regulatorStateSet{\component}$
is covered by some ${\partialRegulatorStateSubset}_{\transitionValueChange} \cup
{\partialRegulatorStateSubset}_{\textsc{ext}}$ can be checked in
$\mathcal{O}(\size{\regulatorSet{\component}})$.
Each $\regulatorState$ requires at most $\size{\regulatorSet{\component}}$
such tests (but usually much less).
As such, the extension set can be computed in
$\mathcal{O}({\size{\regulatorSet{\component}}}^2)$ and thus,
for all the regulator states:
$\mathcal{O}(\size{\regulatorStateSet{\component}} \cdot
{\size{\regulatorSet{\component}}}^2)$.
Finally, the quasilinear complexity comes from the need to keep
the regulator states in a priority queue giving
us the final complexity of
$\mathcal{O}(\size{\regulatorStateSet{\component}} \cdot
(\textsc{log}(\size{\regulatorStateSet{\component}}) +
{\size{\regulatorSet{\component}}}^2))$.

Algorithm~\ref{alg:cover_set_computation} does not require explicit enumeration
of parametrisations when coupled with the parametrisation set semantics $\abstractSemantics$.
The parametrisation set is only used to determine which regulator states enable
the value change (queries to ${\partialRegulatorStateSubset}_{\textsc{ena}}$).
This information is readily available using $\abstractSemantics$
in the form of parameter values of the relevant bound.


\section{Discussion}
\label{sec:discussion}

The goal-oriented model reduction procedure for parametrised models
has been extended to parametric regulatory networks.
The parametric reduction procedure is compatible with a large family
of parametrisation set semantics functions, including the over-approximating
semantics introduced for PRNs in~\cite{me18},
without the need to enumerate the parametrisations explicitly.

The reduction method can be applied alongside the model refinement procedure
based on unfolding~\cite{me18}.
The parametric reduction can be applied on-the-fly within PRN unfolding
in the same fashion the reduction procedure for parametrised networks
is applied in Petri net unfoldings~\cite{Chatain17}.
The application to PRN unfoldings suffers from the same challenge
with cut-off events as the parametrised version with Petri net unfoldings.
The challenge arises from the need to keep track of the transition set
as the model evolves (transitions are pruned) by the reduction procedure
along the unfolding process.
Moreover, a similar challenge is already present in PRN unfoldings due to
parametrisation sets~\cite{me18}.
Two different methods are used to tackle the issue.
In~\cite{Chatain17}, if more transitions are encountered during the unfolding,
the respective branch is reiterated with the new transition set.
In~\cite{me18}, a new branch is introduced into the unfolding for the new
parametrisation set instead.
Both of the methods are applicable for transitions (respectively,
$\activationLimitVector$ and $\inhibitionLimitVector$) in PRN unfoldings
with model reduction.

The parametric reduction is an independent procedure and can be applied
in any other setting besides the mentioned coupling with model refinement.
Moreover, should complexity be a concern, several possibilities to
abstract the procedure exist.
The regulation cover set allows for a different algorithm, or even to relax
the definition itself.
Or, the condition for a trace to be realisable can be dropped from the
validity criterion for objectives to avoid having to check against
parametrisation sets.
Both of the suggested approximations are sound as adding new transitions has
no effect on minimal traces.

Future work includes the refinement of the interplay between parametric model
reduction and model refinement, further applications and extensions of the
parametric model reduction itself and application of goal-oriented reduction
to a wider variety of parametric models.

\bibliographystyle{plainurl}
\bibliography{manuscript}


\begin{appendix}


\section{Proof of Theorem~\ref{thm:minimal_trace_preservation}}
\label{apx:minimal_trace_preservation_proof}

Here we conduct the proof of Theorem~\ref{thm:minimal_trace_preservation}
stating that by conducting reduction of a DPRN $\dprn$ w.r.t. an initial state
$\state$ and goal $\goal$, all transitions of all minimal traces are
preserved by the means of compatible partial transitions in
$\validObjectiveTransitionSet{\state}{\reducedObjectiveSet}$.

We first show that an existence of an objective $\objective$ that covers
a transition $\transition =
\explicitTransition{\component}{\myValue}{\secondValue}{\regulatorState}$
of a realisable trace, formally $\objective =
\explicitObjective{\component}{\myValue - \integer \scalarMultiplication
\transitionSign{\transition}}
{\secondValue + \secondInteger \scalarMultiplication
\transitionSign{\transition}}$
for some $\integer, \secondInteger \in \naturalWithZero$,
in $\reducedObjectiveSet$ is enough to guarantee existence of a compatible
partial transition and thus preservation of the transition.


{ 
\def\someTransition{\vectorElement{\trace}{\index}}

\begin{lemma}
\label{lem:objective_to_partial_transition}

Let a realisable trace $\trace$ of an DPRN $\dprn$ reach goal $\goal$ from initial state
$\state$ and let $\reducedObjectiveSet$ be the objective set constructed by
Definition~\ref{def:reduction_procedure} for the given goal and initial state.
Then for any $\someTransition =
\explicitTransition{\component}{\myValue}{\secondValue}{\regulatorState}$
covered by an objective $\objective \in \reducedObjectiveSet$
there exists a partial transition $\partialTransition =
\explicitTransition{\component}{\myValue}
{\secondValue}{\partialRegulatorState} \in
\validObjectiveTransitionSet{\state}{\reducedObjectiveSet}$
such that $\regulatorState \in \partialRegulatorState$.

\end{lemma}

\begin{proof}

Let $\partialRegulatorState \in
\regulationCoverSet{\explicitTransitionValueChange{\component}{\myValue}
{\secondValue}}$ be arbitrary such that
$\regulatorState \in \partialRegulatorState$.
We know at least one such $\partialRegulatorState$ exists by definition of
regulation cover set (Definition~\ref{def:regulation_cover_set}).

Then, by Definition~\ref{def:objective_transitions}, the corresponding
partial transition $\partialTransition = \explicitTransition{\component}
{\myValue}{\secondValue}{\partialRegulatorState} \in
\objectiveTransitionSet{\objective} \subseteq
\objectiveTransitionSet{\reducedObjectiveSet}$.
Finally, since $\trace$ itself is a witness of the validity of objectives
for all regulators required by $\partialTransition$, $\partialTransition \in
\validObjectiveTransitionSet{\state}{\objective} \subseteq
\validObjectiveTransitionSet{\state}{\reducedObjectiveSet}$.

\end{proof}
}

We now show that if a transition of a realisable trace $\trace$
is not covered by an objective in $\reducedObjectiveSet$,
the trace is not minimal.

Let thus $\vectorElement{\trace}{\index}$
with $\componentProjection{\vectorElement{\trace}{\index}} = \component$
be such a transition,
and let $\transition = \vectorElement{\trace}{\thirdIndex}$ be the last
transition covered by an objective in $\reducedObjectiveSet$
such that $\thirdIndex < \index$
and $\componentProjection{\transition} = \component$, if it exists.
Finally, let ${\transition}' = \vectorElement{\trace}{\secondIndex}$ be the
first transition such that $\index < \secondIndex$,
$\componentProjection{{\transition}'} = \component$
and $\vectorElement{\preset{\suffix{\trace}{\secondIndex}}}{\component} =
\vectorElement{\poset{\prefix{\trace}{\thirdIndex}}}{\component}$
(respectively, $\vectorElement{\preset{\suffix{\trace}{\secondIndex}}}
{\component} = \vectorElement{\state}{\component}$
if $\transition$ does not exist), if it exists.

We now construct a trace ${\trace}'$ by removing all transitions
in $\infix{\trace}{\thirdIndex + 1}{\secondIndex - 1}$ which change the value of
$\component$ from $\trace$,
where $\thirdIndex = 0$ if $\transition$ does not exist
and $\secondIndex = \size{\trace}$ if ${\transition}'$ does not exist.
Since the removed transitions form a loop on the value of $\component$,
respectively, have no causal successors modifying the value of $\component$
if ${\transition}'$ does not exist, the evolution of $\component$
along ${\trace}'$ remains valid.

Moreover, if any transition $\vectorElement{\trace}{\kappa}$
covered by $\objective \in \reducedObjectiveSet$
such that $\thirdIndex < \kappa < \secondIndex$
and $\componentProjection{\vectorElement{\trace}{\kappa}} = \component$
exists, then there exists another transition $\vectorElement{\trace}{\iota} \in
\prefix{\trace}{\thirdIndex}$ with the same value change as
$\vectorElement{\trace}{\kappa}$.

Let thus $\explicitObjective{\component}{\alpha}{\secondValue} \in
\reducedObjectiveSet$ cover at least one transition modifying the value of
$\component$ in $\infix{\trace}{\thirdIndex + 1}{\secondIndex - 1}$.
For such an objective to be included there must exist a covered transition requiring value $\secondValue$
of component $\component$ somewhere along the trace and thus,
by point (2) of Definition~\ref{def:reduction_procedure},
$\objectiveValues{\vectorElement{\state}{\component}}{\secondValue} \in
\reducedObjectiveSet$.
The transition $\transition$ therefore has to exist, meaning an objective
$\objectiveValues{\beta}{\vectorElement{\poset{\prefix{\trace}{\thirdIndex}}}
{\component}} \in \reducedObjectiveSet$ also exists.
By point (3) of Definition~\ref{def:reduction_procedure}, we get
$\objectiveValues{\vectorElement{\poset{\prefix{\trace}{\thirdIndex}}}
{\component}}{\secondValue}, \objectiveValues{\secondValue}
{\vectorElement{\poset{\prefix{\trace}{\thirdIndex}}}{\component}} \in
\reducedObjectiveSet$.
Then, since $\vectorElement{\trace}{\index}$ is not covered,
we have $\objectiveSign{\objectiveValues{\beta}
{\vectorElement{\poset{\prefix{\trace}{\thirdIndex}}}{\component}}} =
\objectiveSign{\explicitObjective{\component}{\alpha}{\secondValue}}$
and moreover $\alpha$ is an intermediate value in the objective
$\objectiveValues{\beta}{\vectorElement{\poset{\prefix{\trace}{\thirdIndex}}}
{\component}}$.
Thus, the same value change as the covered transition in
$\infix{\trace}{\thirdIndex + 1}{\secondIndex - 1}$ had to occur in
$\prefix{\trace}{\thirdIndex}$ in order to reach
$\vectorElement{\poset{\prefix{\trace}{\thirdIndex}}}{\component}$.

Since $\objectiveValues{\vectorElement{\state}{\goalComponent}}{\goalValue} \in
\reducedObjectiveSet$ by point (1) of Definition~\ref{def:reduction_procedure},
${\trace}'$ surely reaches the goal from the initial state.
Furthermore, thanks to properties of the parametrisation set semantics
and since $\trace$ is realisable, also
$\admissibleSet{\collapse{{\trace}'}} \neq \emptyset$.
If ${\trace}'$ is valid (regulator state of each transition matches the source state),
$\trace$ is not minimal and we are done.
Let us therefore assume there exists a transition
$\vectorElement{\trace}{\fourthIndex}$ with
$\componentProjection{\vectorElement{\trace}{\fourthIndex}} \neq \component$
which causally depends on one of the removed transitions.
We first show such $\vectorElement{\trace}{\fourthIndex}$ cannot be
covered by an objective in $\reducedObjectiveSet$ by contradiction.

Let thus $\objective \in \reducedObjectiveSet$ cover
$\vectorElement{\trace}{\fourthIndex}$
and let $\secondValue$ be the value of $\component$
$\vectorElement{\trace}{\fourthIndex}$ depends on.
Surely $\secondValue \neq \vectorElement{\poset{\prefix{\trace}{\thirdIndex}}}{\component}$,
as otherwise the removal of $\infix{\trace}{\thirdIndex + 1}{\secondIndex - 1}$
has no effect on $\vectorElement{\trace}{\fourthIndex}$.
As such, since $\objectiveValues{\vectorElement{\state}{\component}}
{\secondValue} \in \reducedObjectiveSet$, $\transition$ must exist.
Then also objective $\objectiveValues{\alpha}
{\vectorElement{\poset{\prefix{\trace}{\thirdIndex}}}{\component}} \in
\reducedObjectiveSet$ and by point (3) of
Definition~\ref{def:reduction_procedure} $\objectiveValues
{\vectorElement{\poset{\prefix{\trace}{\thirdIndex}}}{\component}}
{\secondValue} \in \reducedObjectiveSet$.
Therefore the first transition reaching the value $\secondValue$ of $\component$ in
$\infix{\trace}{\thirdIndex}{\fourthIndex}$ is covered, which is a contradiction
with $\vectorElement{\trace}{\thirdIndex}$ being the last covered transition.

We can thus repeat the reasoning for the uncovered $\vectorElement{\trace}{\fourthIndex}$
and remove the respective loop or unused tail of evolution of component
$\componentProjection{\vectorElement{\trace}{\fourthIndex}}$.
Since $\trace$ is finite, all invalid transitions are ultimately
removed while retaining realisability and reachability of the goal.
$\trace$ is thus not minimal.

\end{appendix}

\end{document}